\documentclass[11pt]{article}
\usepackage{geometry}
\usepackage{amsmath, amssymb, amsthm}
\usepackage{graphicx}
\usepackage{bm}
\usepackage[linesnumbered,ruled]{algorithm2e}
\newtheorem{theorem}{Theorem}

\newtheorem{corollary}{Corollary}
\newtheorem{claim}{Claim}

\usepackage{cleveref}

\graphicspath{{./figure/}}%

\crefname{claim}{Claim}{Claims}
\crefname{theorem}{Theorem}{Theorems}
\title{Reforming an Envy-Free Matching\thanks{%
A preliminary version will appear in Proceedings of the 36th AAAI Conference on Artificial Intelligence (AAAI 2022).  
This work was supported by 
JSPS KAKENHI Grant Numbers JP18H04091, JP19K11814, JP20H05793, JP20H05795, JP20K11670, JP20K23323, JP18H05291, JP19H05485, JP21H03397.}}
\author{Takehiro Ito%
\thanks{Tohoku University, {\tt takehiro@tohoku.ac.jp}}
\and 
Yuni Iwamasa%
\thanks{Kyoto University, {\tt iwamasa@i.kyoto-u.ac.jp}}
\and 
Naonori Kakimura%
\thanks{Keio University, {\tt kakimura@math.keio.ac.jp}}
\and 
Naoyuki Kamiyama%
\thanks{Kyushu University, {\tt kamiyama@imi.kyushu-u.ac.jp}}
\and 
Yusuke Kobayashi%
\thanks{Kyoto University, {\tt yusuke@kurims.kyoto-u.ac.jp}}
\and 
Yuta Nozaki%
\thanks{Hiroshima University, {\tt nozakiy@hiroshima-u.ac.jp}}
\and 
Yoshio Okamoto%
\thanks{The University of Electro-Communications, {\tt okamotoy@uec.ac.jp}}
\and
Kenta Ozeki%
\thanks{Yokohama National University, {\tt ozeki-kenta-xr@ynu.ac.jp
}}
}
\date{}

\begin{document}

\maketitle

\begin{abstract}
We consider the problem of reforming an envy-free matching when
each agent is assigned a single item. Given an envy-free matching,
we consider an operation to exchange the item of an agent with 
an unassigned item preferred by the agent that results in another 
envy-free matching. 
We repeat this operation as long as we can. We prove that
the resulting envy-free matching is uniquely determined 
up to the choice of an initial envy-free matching, and can be 
found in polynomial time. We call the resulting 
matching a reformist envy-free matching, and then we 
study a shortest sequence to obtain the reformist envy-free matching 
from an initial envy-free matching. We prove 
that a shortest sequence is computationally hard to 
obtain even when each agent accepts at most four items and each item is 
accepted by at most three agents. On the
other hand, we give polynomial-time algorithms when 
each agent accepts at most three items or 
each item is accepted by at most two agents. 
Inapproximability and fixed-parameter (in)tractability
are also discussed.
\end{abstract} 

\section{Introduction}

\emph{Matching under preferences} constitutes an important
and 
well investigated subarea of 
economics and game theory,
and its computational aspects are intensively studied in 
algorithmic game theory and computational social choice
(see, e.g., \cite{KlausMR16,M13}).
In a lot of situations, we are interested in allocating 
indivisible
items, namely, items that cannot be subdivided into several 
parts.
Examples include job allocation, college admission, school 
choice,
kidney exchange, and junior doctor allocation to hospital 
posts. 
Especially, 
this paper is concerned with the situation where each agent 
is assigned a single item.
This situation is often 
called the \emph{house allocation problem}.
A set of agents faces a set of items, 
and each agent has a preference over her acceptable items 
(i.e., her preference list can be incomplete). 
In this situation, 
there may be many possible matchings. 
However, some of those matchings suffer from 
``instability.''

Stability is often studied in terms of \emph{envy} of 
agents in the house allocation problem. 
Given a matching, an agent $i$ has a (justified) 
envy for another agent $j$ if the agent $i$ prefers 
the item assigned to $j$ to the item assigned to $i$. 
If there is no agent with envy, 
the matching is said to be envy-free.
Even with envy-freeness, there may be many 
possible matchings, and we want to look 
for a good envy-free matching.
This motivates the following simple procedure that can be
implemented in a decentralized way. 
Agents start with any envy-free matching. 
There are many unassigned items on the table. 
Then an agent $i$ can exchange the item $x$ 
assigned to her with an item $y$ on the 
table if $i$ prefers $y$ to $x$ and 
the exchange does not break the envy-freeness. 
This ``reforming'' process can continue until 
no agent has an incentive for exchange. 
Then every agent will be assigned an item 
that is at least as good as the item that 
was initially assigned, and the 
resulting matching is still envy-free.

Our problem arises in the following situation. 
First, items are assigned to agents by an envy-free matching.
The matching is given \emph{a priori}, and 
agents are satisfied by the items assigned to them.
Then the agents face the arrival of extra items.
This may happen, for example, when some new items are brought into the market, or when some of the agents leave the
market and release their items.
Since the new items could 
improve agents' utilities, the agents might not be satisfied with
the items currently assigned to them any longer.
Hence, we want to reassign items by incorporating the existence of new items. 
One way to redistribute items is to compute a 
new envy-free matching 
from scratch. However, this requires 
the agents first to release their items,
which will results in the decrease of their utilities.
Our proposal here is to exchange items one by one so that
the intermediate matchings are all envy-free and
no agent decreases her utility at any moment during the procedure.

In this paper, we call a matching obtained by the process above a 
\emph{reformist envy-free matching}.
A reformist envy-free matching can 
depend on the choice of an initial envy-free matching
and the
sequence of exchanges.
Our first result states that the exchange sequence 
does not affect the resulting reformist envy-free matching.
Namely, a reformist envy-free matching 
uniquely exists up to the
choice of an initial envy-free matching (\cref{thm:unique}).

The definition of a reformist envy-free matching 
was motivated by a decentralized algorithm.
However, the number of steps in this process is not discussed yet.
With a decentralized algorithm, we may end up with an extremely 
long sequence
of envy-free matching until we obtain a reformist envy-free matching.
On the other hand, if there is coordination among the agents, 
they may quickly
obtain a reformist envy-free matching.
Coordination is modeled as a centralized 
algorithm in which a central authority
declares who should exchange an item next, and agents obey the declarations
of the central authority.
Since a reformist envy-free matching is unique (\cref{thm:unique}), 
there is no reason for agents to deviate from the orders of the central authority.

To formalize the discussion, we consider the following type of algorithms.
Until a reformist envy-free matching is obtained,
an agent is nominated at each step.
Let $i$ be the nominated agent. Then $i$ exchanges the currently assigned item with an unassigned item on the table that is 
most preferred by $i$ such that the matching 
after the exchange is still envy-free.
The choice of nominated agents can change the number of steps.
In the decentralized setting the choice will be done arbitrarily while
in the centralized setting the choice is supposed to be done cleverly
to minimize the number of steps.
Thus, we examine the minimum number of steps to obtain a reformist envy-free matching 
with respect to a given initial envy-free matching.

In what follows, we call a sequence of exchanges to obtain the 
reformist envy-free matching a {\em reformist sequence}, and 
we call the problem of finding a shortest reformist sequence the 
{\em shortest reformist sequence problem}. 
We define the decision version of 
the shortest reformist 
sequence problem as the problem 
where we are given an envy-free matching
$\mu$ and a positive integer $\ell$, and 
we determine whether there is a reformist sequence of length at most 
$\ell$ with respect to the initial envy-free matching $\mu$.
To justify the study of the shortest reformist sequence problem, 
we first show that coordination sometimes makes sense by giving an
example 
in which the maximum number of steps can be arbitrarily
larger than the
minimum number of steps (\cref{thm:exponential-example}).
Then, we prove that the decision version of the shortest reformist
sequence problem is NP-complete even if each agent accepts at most
four items (i.e., the preference list of each agent contains at most
four items) and each item appears in the preference lists of at most
three agents (\cref{thm:NPc}).
On the other hand, the shortest reformist sequence problem can be
solved in polynomial time if each agent accepts at most three items
(\cref{thm:atmost3items}) \emph{or} each item appears in the
preference lists of at most two agents (\cref{thm:<=2}).

With the NP-completeness result, we consider two established
approaches to cope with NP-completeness, namely approximation and
fixed-parameter tractability.
For approximation, we indeed prove that the shortest reformist
sequence problem is hard to 
approximate within the factor of $c \ln n$
for some constant $c$, where $n$ is the number of agents
(\cref{thm:inapx}).
For fixed-parameter tractability, we have several choices of
parameters.
When the length $\ell$ of a reformist sequence is chosen as a
parameter, (the decision version of) the shortest reformist sequence
problem is fixed-parameter tractable (\cref{thm:fpt2}).
On the other hand, When $\ell - n$ is chosen as a parameter, the
problem is W[1]-hard (\cref{thm:w1hard}), where $n$ is the number of
agents.
The choice of the parameter comes from the property that the length of
a reformist sequence is at least $n$ after preprocessing and thus the
parameter is considered the number of redundant steps in the reformist
sequence.
On the other hand, when the number of ``intermediate'' items is chosen
as a parameter, the problem is fixed-parameter tractable
(\cref{thm:fpt}). 
Here, ``intermediate'' items are items that are not assigned in the
initial envy-free matching 
or in the reformist envy-free matching.

\paragraph{Related Work}

The concept of envy-freeness is often used in the literature of 
social choice theory.
For example, Gan, Suksompong, and Voudouris~\cite{GSV19} 
considered the problem of checking the existence of 
an envy-free item matching in the situation where 
any agent accepts all the items and the 
preferences may contain ties. 
They proved that 
we can determine whether there is an 
envy-free item matching 
in polynomial time.
Beynier et al.~\cite{BCGLMW18} considered 
envy-freeness on an envy
relationship network. 
Envy-freeness is also studied
in the literature on fair division of divisible goods such as cake cutting
(e.g.~\cite{Procaccia16,AzizM20,GoldbergHS20}),
on fair division of indivisible goods with numerical valuations (e.g.~\cite{BouveretCM16,ChaudhuryGM20}), and
in two-sided markets
such as the hospitals/residents problem
(e.g.~\cite{WuR18,Yokoi20,KrishnaaLNN20}).

Problems of improving a given item allocation via some operations 
have been considered in the study of item allocations. 
Gourv\`{e}s, Lesca, and Wilczynski~\cite{GLW17}
considered the problem of determining whether 
a target item allocation can be reached via rational swaps 
on a social network.
Furthermore, they considered that the problem of 
determining whether some specified agent can get a target 
item via rational swaps
(see also \cite{BW19,HX20}). 

Our problems are closely related to 
the study of {\em combinatorial reconfiguration}. 
In combinatorial reconfiguration, we consider 
problems where we are given an initial configuration and 
a target configuration of some combinatorial objects, and 
the goal is to check the reachability between these two 
configurations via some specified operations. 
The study of algorithmic aspects of combinatorial reconfiguration 
was initiated in \cite{IDHPSUU11}.
See, e.g., \cite{N18} for a survey of combinatorial 
reconfiguration. 

\section{Preliminaries}

Throughout this paper, a finite set of $n$ 
\emph{agents} is denoted by $N$,
and a finite set of $m$ \emph{items} is denoted by $M$.
Each agent $i \in N$ is associated with a subset $M_i \subseteq M$ 
and a strict total order $\succ_i$ on $M_i$:
$M_i$ represents the set of acceptable items for $i$, and
$\succ_i$ represents the preference of $i$ over $M_i$.
For each agent $i \in N$, we define $m_i := |M_i|$. 
For each agent $i \in N$, if $M_i = \{x_1, x_2, \dots, x_{m_i}\}$ and 
$x_1 \succ_i x_2 \succ_i \dots \succ_i x_{m_i}$, then 
we describe $\succ_i$ by 
$\succ_i \colon x_1, x_2, \dots, x_{m_i}$.
For each agent $i \in N$ and each pair $x,y \in M$ of items, 
we write $x \succeq_i y$ if $x\succ_i y$ or $x=y$.
Note that $\succ_i$ satisfies transitivity, i.e., if $x \succ_i y$ and
$y \succ_i z$, then $x \succ_i z$.

An injective mapping $\mu\colon N \to M$ is called a \emph{matching}
if $\mu(i) \in M_i$ for every agent $i \in N$.
For each matching $\mu$,
an item $x \in M$ is \emph{assigned} if there exists an 
agent $i \in N$
such that $\mu(i)=x$; otherwise $x$ is \emph{unassigned}.
A matching $\mu$ is \emph{envy-free} if there exists no 
pair $i,j \in N$ of distinct 
agents such that
$\mu(j) \succ_i \mu(i)$.
For each matching $\mu$, 
we denote the set of unassigned items for $\mu$ by $\overline{M}_\mu$.

Let $\mu, \sigma$ be envy-free matchings.
We write 
$\mu \leadsto \sigma$ if there exists an agent 
$i \in N$ with the following two conditions:
(i) $\sigma (i) \succ_i \mu(i)$;
(ii) $\mu(j) = \sigma(j)$ for every agent $j \in N \setminus \{i\}$.
Intuitively, if items are assigned to the agents 
according to $\mu$ and 
$\mu \leadsto \sigma$, then $\sigma(i) \in \overline{M}_\mu$ and 
$i$ has an incentive to exchange her item $\mu(i)$ with $\sigma(i)$ 
and the resulting matching is still envy-free.
This way, the operation ``$\leadsto$'' unilaterally 
improves the current 
envy-free matching $\mu$ to a new envy-free matching $\sigma$.

Let $\mu, \sigma$ be envy-free matchings.
If there exist envy-free matchings 
$\mu_0, \mu_1, \dots, \mu_{\ell}$ such that
(1) $\mu_0 = \mu$, $\mu_{\ell} = \sigma$,
(2) $\mu_t \leadsto \mu_{t+1}$ for every integer
$t \in \{0,1,\dots,{\ell}-1\}$, and
(3) there exists no envy-free matching $\mu^{\prime}$ 
such that $\mu_{\ell} \leadsto \mu^{\prime}$,
then $\sigma$ is called a \emph{reformist envy-free matching}
with respect to $\mu$.
Intuitively, a reformist envy-free matching with respect 
to $\mu$ is an envy-free
matching that is obtained from $\mu$ as an outcome of 
the iterative improvement.

To illustrate envy-free matchings, 
we introduce a graph representation.
Given a set $N$ of agents, a set $M$ of items, $M_i$ and 
$\succ_i$ for all agents
$i \in N$, we create the following directed graph.
The vertex set is $M$, the set of items.
For each agent $i \in N$ with $M_i = \{x_1,x_2,\dots,x_k\}$ and 
$\succ_i \colon x_1, x_2, \dots, x_k$, we place $k-1$ arcs 
$(x_2,x_1), (x_3,x_2), \dots, (x_{k}, x_{k-1})$:
those arcs are labeled by
$i$.
There can be parallel arcs from $x$ to $y$ with different labels,
or they can be identified with a single arc from $x$ to $y$ 
with multiple labels.

An example is given in \figurename~\ref{fig:example1-1}.
There are four agents $1, 2, 3, 4$ and seven items $a, b, c, d, e, f, g$. The preferences are given as follows:
\begin{align*}
    \succ_1 \colon& a, b, c, d, e, f, g;\\
    \succ_2 \colon& f, d, a, g, e;\\
    \succ_3 \colon& b, g, a, c;\\
    \succ_4 \colon& d, c, g, e, f.
\end{align*}
The colors are assigned for agents: 
black for agent $1$, blue for agent $2$,
red for agent $3$, and violet for agent $4$.

\begin{figure}[ht]
\centering
\includegraphics[page=1, width = 8cm]{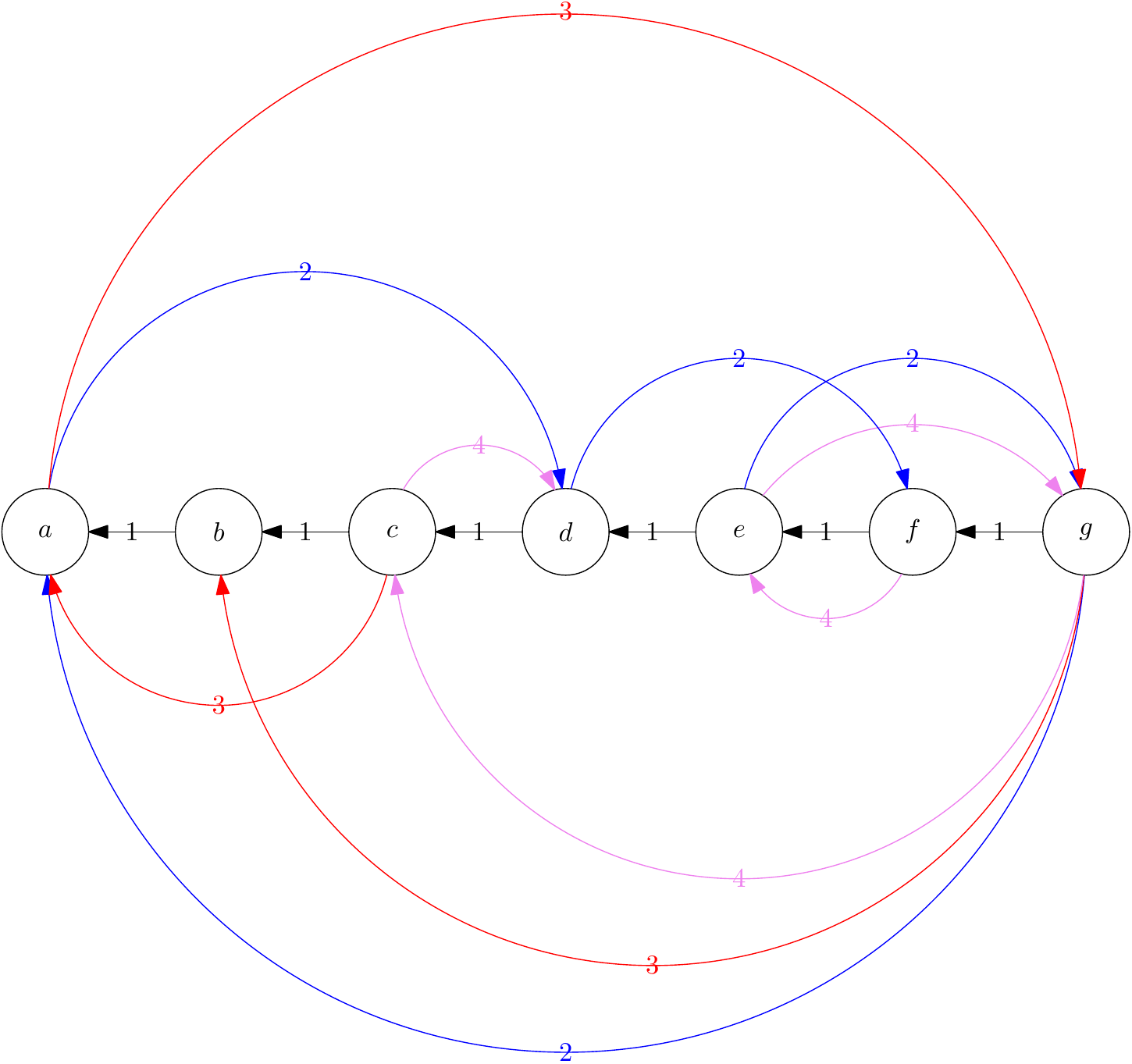}
\caption{A graph representation of an instance.}
\label{fig:example1-1}
\end{figure}

A matching $\mu$ is identified with a labeled token placement.
A token for each agent $i$ is placed on the vertex $\mu(i)$:
the token is labeled by $i$, and for convenience we 
denote the token by $i$.
Since $\mu$ is a matching, no vertex holds two or more tokens.
If a matching $\mu$ is envy-free, then there exist no pair 
of tokens $i,j$ such
that $j$ is placed on a vertex that can be reached from 
the vertex holding $i$ 
along arcs labeled by $i$; the converse also holds.
The operation $\mu \leadsto \sigma$ corresponds to 
moving the token at $\mu(i)$ to $\sigma(i)$.
Labels are often identified with colors in our figures.

In \figurename~\ref{fig:example1-2},
labeled tokens are placed on vertices.
The labels of tokens are shown by colors.
The token $1$ (black) is placed at vertex $b$,
the token $2$ (blue) is placed at vertex $d$,
the token $3$ (red) is placed at vertex $g$,
and the token $4$ (violet) is placed at vertex $e$.
In this example, agent $4$ has an envy for
agent $3$ since the token $3$ is placed on the vertex $g$
that can be reached from the vertex $e$ holding $4$
along arcs labeled by $4$ (i.e., violet arcs).
Similarly, agent $3$ has an envy for agent $1$.

\begin{figure}[ht]
\centering
\includegraphics[page=2, width = 8cm]{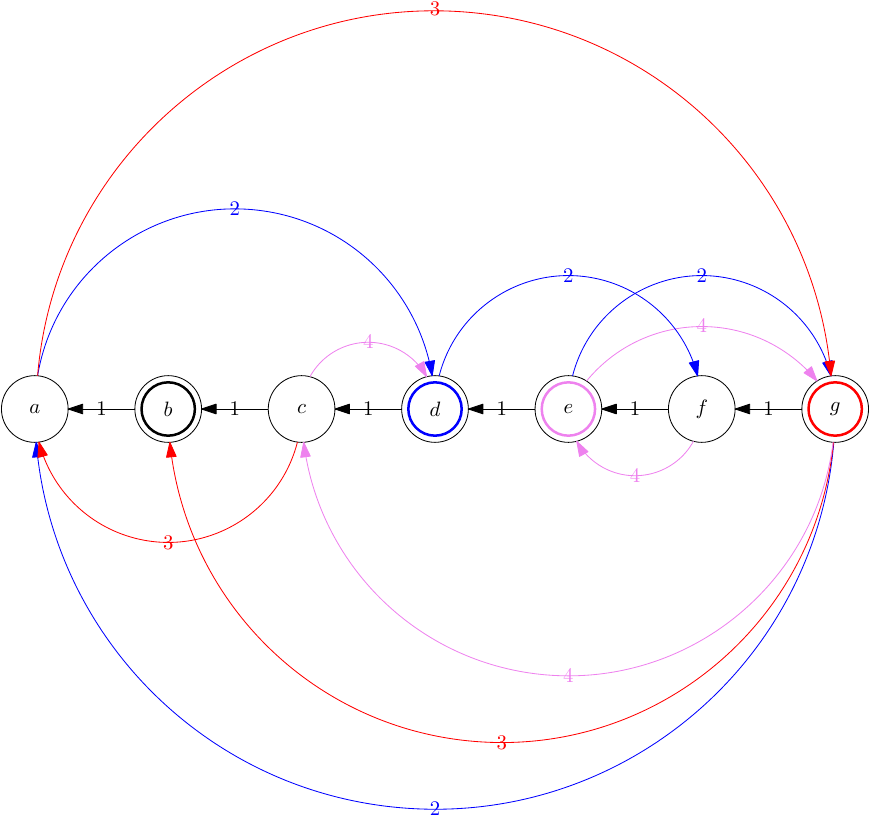}
\caption{A graph representation of a matching with envy.}
\label{fig:example1-2}
\end{figure}

We conclude this section with a small example.
\paragraph{Example}
Consider $2$ agents $N=\{1,2\}$ and $5$ items $M = \{x, y, p, q, r\}$ with preferences 
\begin{align*}
\succ_1 : p, r, q, x \quad \text{and}\quad \succ_2 : q, p, y.
\end{align*}
See \figurename~\ref{fig:smallexample1}.
Let $\mu$ be a matching satisfying $\mu(1)=x$ and $\mu (2) = y$.
Then it is confirmed to be envy-free.
However, in the matching $\mu$, agent $1$ has an incentive to exchange her current item $x$ with $r$, and such exchange does not arouse envy in agent $2$.
Thus we can improve $\mu$ to $\mu_1$, where $(\mu_1 (1), \mu_1(2)) = (r, y)$, which we denote by $\mu \leadsto \mu_1$. 
Similarly, we have $\mu_1 \leadsto \mu_2\leadsto \mu_3$, where $(\mu_2 (1), \mu_2(2)) = (r, q)$ and $(\mu_3 (1), \mu_3(2)) = (p, q)$.
Since $p$ and $q$ are the most preferred items for the agents, $\mu_3$ is the reformist envy-free matching.

\begin{figure}[ht]
\centering
\includegraphics[width=\textwidth]{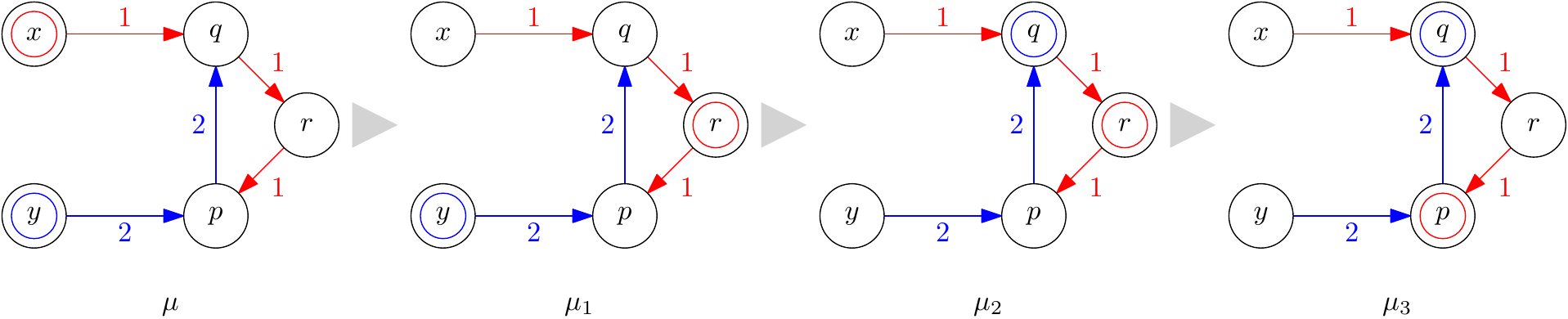}
\caption{A small example.}
\label{fig:smallexample1}
\end{figure}

\section{Uniqueness}\label{sec:uniqueness}

We first observe that a reformist envy-free matching with respect to an envy-free matching can be obtained in polynomial time.
In fact, since one exchange strictly improves the current matching, the number of exchanges to obtain a reformist envy-free matching is at most $|M|\cdot |N|$.

We prove that the obtained reformist envy-free matching is unique up to the choice of an initial envy-free matching.

\begin{theorem}
\label{thm:unique}
Let $\mu$ be an envy-free matching.
Then a reformist envy-free matching with respect 
to $\mu$ uniquely exists.
\end{theorem}
\begin{proof}
The existence is immediate from the definition. 
We prove the uniqueness.
Suppose to the contrary that there exist reformist envy-free matchings 
$\sigma$ and $\tau$ with 
respect to $\mu$ such that $\sigma \neq \tau$. 
Without loss of generality, we can assume that 
there exists an agent $i \in N$ such that 
$\sigma(i) \succ_{i} \tau(i)$. 
Suppose that for envy-free matchings
$\sigma_0, \sigma_1, \dots, \sigma_\ell$, we have 
$\mu = \sigma_0 \leadsto \sigma_1 \leadsto \dots
\leadsto \sigma_{\ell} = \sigma$. 
Since $\tau$ is a reformist envy-free 
matching with respect to $\mu$, 
$\tau(j) \succeq_j \sigma_0(j)$ holds 
for every agent $j \in N$. 
Let $t$ be the minimum integer in $\{1,2,\dots,\ell\}$ such that 
$\sigma_t(i) \succ_{i} \tau(i)$ for some agent $i \in N$. 
Then $\tau(j) \succeq_j \sigma_t(j)$ holds
for every agent $j \in N \setminus \{i\}$. 

If there is an agent $j \in N \setminus \{i\}$ such that 
$\tau(j) = \sigma_t(i)$, then 
$\tau(j) \succ_i \tau(i)$, which contradicts the assumption that $\tau$ is envy-free. 
This implies that $\tau(j)\neq \sigma_t(i)$ holds for 
every agent $j \in N \setminus \{i\}$, 
which means $\sigma_t(i) \in \overline{M}_\tau$.
Hence, under the matching $\tau$, the agent $i$ can exchange
$\tau (i)$ with $\sigma_t(i)$ to obtain another matching
$\tau^{\prime}$.
Since $\tau$ is a reformist envy-free matching, 
the resulting matching $\tau^{\prime}$ is not envy-free. 
That is, there is an agent $j \in N \setminus \{i\}$ 
such that 
$\tau^{\prime}(i) \succ_j \tau^{\prime}(j) = \tau(j)$. 
For such an agent $j \in N \setminus \{i\}$,
we have 
$\sigma_t(i) = \tau^{\prime}(i) \succ_j \tau(j) \succeq_j \sigma_t(j)$.
However, this means that 
the agent $j$ has envy for $i$ on $\sigma_t$, which 
contradicts the fact that $\sigma_t$ is envy-free.
This completes the proof. 
\end{proof}

Theorem \ref{thm:unique} has a consequence for the following reconfiguration question.
Namely, we are given two envy-free matchings $\mu$ and $\tau$, and 
asked to determine whether
$\tau$ is obtained from $\mu$ by the iterative improvement. 
\begin{corollary}
\label{cor:reconfig}
For two envy-free matchings $\mu,\tau$,
we can determine whether 
there exists a sequence of envy-free matchings 
$\mu=\mu_0, \mu_1, \dots, \mu_{\ell}=\tau$ such that 
$\mu_t \leadsto \mu_{t+1}$ for every integer 
$t \in \{0,1,\dots,{\ell}-1\}$. 
\end{corollary}
\begin{proof}
If there is an agent $i$ such that $\mu(i) \succ_i \tau(i)$, then
the answer is No.
Therefore, assume that $\tau(i) \succeq_i \mu(i)$ for every agent $i$.

The algorithm first removes each item $x$ from the instance if $x \succ_i \tau(i)$ for some agent $i\in N$. 
Then it computes the reformist envy-free matching $\sigma$ with respect to $\mu$.
If $\sigma = \tau$, then we know 
$\tau$ is reached from $\mu$ 
and the answer is Yes.
Otherwise (i.e., $\sigma \neq \tau$), there exists an agent $i$ such that $\tau(i) \succ_i \sigma(i)$ since all the items $x$ with $x \succ_i \tau(i)$ were already removed from the instance. 
Since the reformist envy-free matching with respect to $\mu$ is unique (Theorem \ref{thm:unique}), $\tau$ cannot be reached from $\mu$, and the answer must be No.
\end{proof}

\section{Shortest Reformist Sequence: Hardness}

To justify the study of the shortest reformist sequence problem,
we first give an example in which the maximum length of a reformist
sequence can be arbitrarily larger than the minimum length.

\begin{theorem}
\label{thm:exponential-example}
 For any positive integer $p$,
 there is an instance of the shortest reformist sequence 
 problem with three agents and $2p+3$ items such that there is 
 a reformist sequence of length $2p-1$ while the shortest reformist
 sequence has length at most four.
\end{theorem}
\begin{proof}
 We construct a desired instance as follows.
 Let $N=\{1, 2, 3\}$ and 
 \begin{equation*}
 M=\{a_\ell, b_\ell \mid \ell=1,2,\dots, p\}\cup\{r, s, z\}.
 \end{equation*}
 We define the preferences of the $3$ agents as follows.
 \begin{align*}
   \succ_{1} &\colon a_p, b_p, a_{p-1}, b_{p-1}, \dots, a_2, b_2, a_1,\\
   \succ_{2} &\colon b_p, z, b_{p-1}, a_{p}, b_{p-2}, a_{p-1}, \dots, b_2, a_3, b_1,\\
   \succ_{3} &\colon r, z, s.
 \end{align*}
 We define the initial matching $\mu$ to be $\mu (1)=a_1$, $\mu (2) = b_1$, and $\mu (3)=s$.
 Then the reformist matching $\sigma$ with respect to $\mu$ is $\sigma(1)=a_p$, $\sigma(2)=b_p$, and $\sigma(3)=r$.
 See Figure~\ref{fig:expo}.

 We observe that we can reach $\sigma$ in four steps as follows: the agent $3$ exchanges $s$ with $r$, 
 the agent $2$ exchanges $b_1$ with $z$, the agent $1$ exchanges $a_1$ with $a_p$, and then the agent $2$ exchanges $z$ with $b_p$. 
 See Figure~\ref{fig:badexample1_opt}.
 
 On the other hand, if the agent $3$ is nominated after the agents $1$ and $2$, the number of steps to reach $\sigma$ is $2p-1$
 (see Figure~\ref{fig:badexample1_long}).
 In the beginning, only the agent $1$ can be nominated to exchange $a_1$ with $a_2$.
 Since $b_2$ receives no envy from the agent $1$ after the exchange, the agent $2$ can exchange $b_1$ with $b_2$.
 Then, $a_3$ has no envy from the agent $2$, implying that the agent $1$ can exchange $a_2$ with $a_3$.
 In such a way, for an integer $i \in \{1,2,\dots, p-1\}$,
 the $(2i-1)$-st step exchanges $a_i$ with $a_{i+1}$ for the agent $1$, and the $2i$-th step exchanges $b_i$ with $b_{i+1}$ for the agent $2$.
 In the end, the two agents reach $a_p$ and $b_p$.
 This transformation is unique, and the number of necessary steps is $2p-2$.
 Finally, the agent $3$ exchanges $s$ with $r$.
 Thus the total number of steps is $2p-1$.
\end{proof}

\begin{figure}[ht]
\centering
\includegraphics{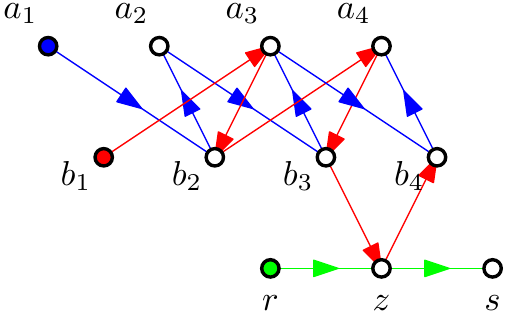}
\caption{Instance in the proof of  Theorem~\ref{thm:exponential-example} with $p=4$.
Colors represent labels, and colored vertices correspond to the items assigned to agents in the initial matching $\mu$.}
\label{fig:expo}
\end{figure}

\begin{figure*}[ht]
\centering
\resizebox{0.9\textwidth}{!}{\includegraphics{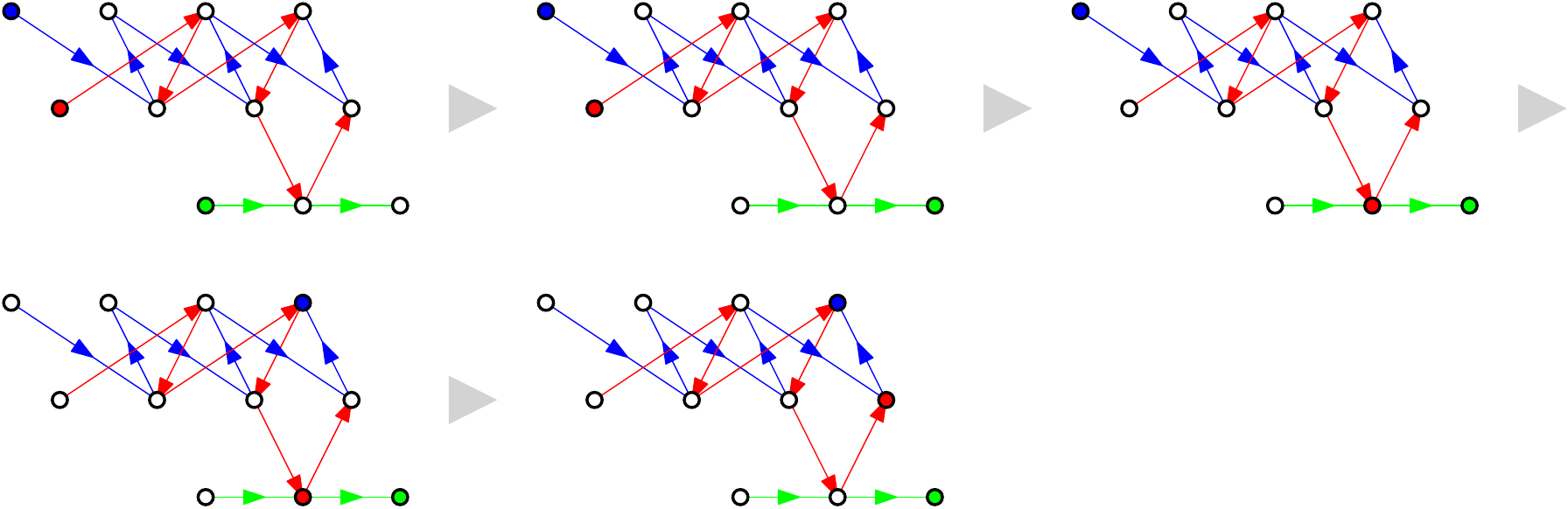}}
\caption{Shortest sequence for the instance in the proof of Theorem \ref{thm:exponential-example}.}
\label{fig:badexample1_opt}
\end{figure*}

\begin{figure*}[ht]
\centering
\resizebox{0.9\textwidth}{!}{\includegraphics{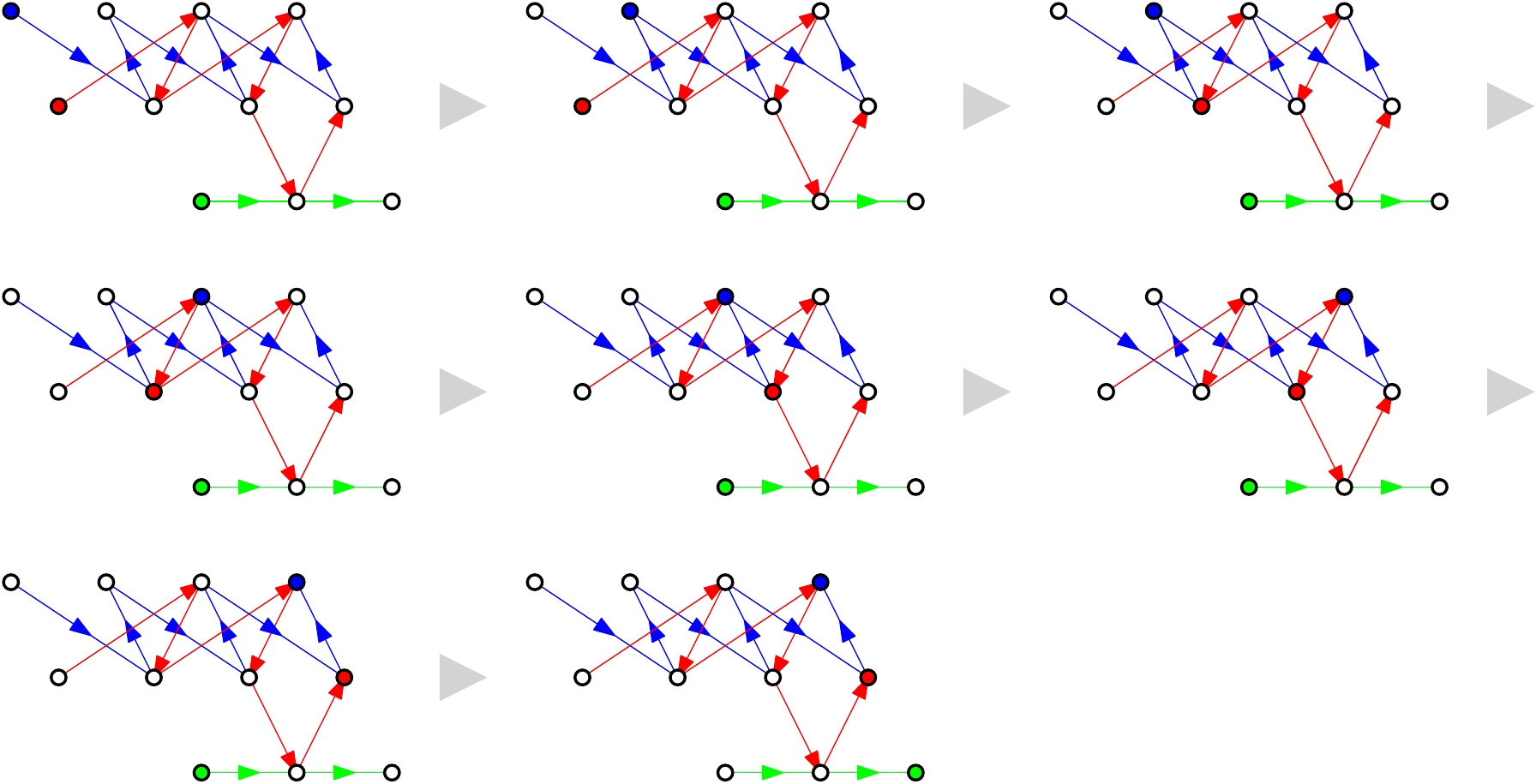}}
\caption{A long sequence for the instance in the proof of Theorem \ref{thm:exponential-example}.}
\label{fig:badexample1_long}
\end{figure*}

As it turns out, (the decision version of) the shortest reformist sequence problem is NP-complete.

\begin{theorem}\label{thm:NPc}
The decision version of the shortest reformist 
sequence problem is {\rm NP}-complete even when $m_i \leq 4$ 
for every agent $i\in N$ and 
$|\{i \in N \mid x \in M_i\}| \le 3$ for every item $x \in M$.
\end{theorem}
\begin{proof}
We first observe that the problem is in NP\@.
This is because one exchange strictly improves the current matching, and hence the maximum number of exchanges in the reformist sequence is at most $|M|\cdot |N|$.

We reduce the vertex cover problem in $3$-regular graphs to the decision version of the 
shortest reformist sequence problem.
In the vertex cover problem, we are given an 
undirected graph $G=(V, E)$ 
and a positive integer $k$, and we are 
asked to determine whether $G$ has 
a subset $S\subseteq V$ such that $|S| \le k$ and 
every edge $e \in E$ has one of its endvertices in $S$
(i.e., $S\cap e \neq \emptyset$).
Such a vertex subset $S$ is called a vertex cover.
It is known~\cite{K72} that the vertex cover problem is 
NP-complete even when a given graph is $3$-regular.
Let $G=(V, E)$ be a $3$-regular graph as an instance of the vertex cover problem.

We construct an instance of the decision version of the shortest reformist 
sequence problem as follows
(see Figure \ref{fig:npc1.pdf}).
For each edge $e \in E$, we prepare four agents 
$e^1,e^2,e^3,e^4$, and
for each vertex $v\in V$, 
we prepare eight agents 
$v^1,v^2,\dots,v^8$.
Thus, there are $4|E|+8|V|$ agents: 
\begin{equation*}
N :=\{e^\ell \mid e\in E, \, \ell \in \{1,2,3,4\}\} \cup 
\{v^\ell\mid v\in V, \, \ell \in \{1,2,\dots,8\}\}.
\end{equation*}
We set 
\begin{equation*}
M :=\{r_i, s_i\mid i\in N\} 
 \cup \{t_v\mid v\in V\}\cup \{y_{e, u}, y_{e, v}\mid e=\{u,v\}\in E\} 
 \cup \{x_{v, e} \mid v\in V, e\in \delta(v)\},
\end{equation*}
where $\delta (v)$ denotes the set of edges incident to $v$.
Note that $|\delta (v)|=3$ for every vertex $v\in V$ since $G$ is $3$-regular.

\begin{figure*}[ht]
\centering
\resizebox{0.9\textwidth}{!}{\includegraphics{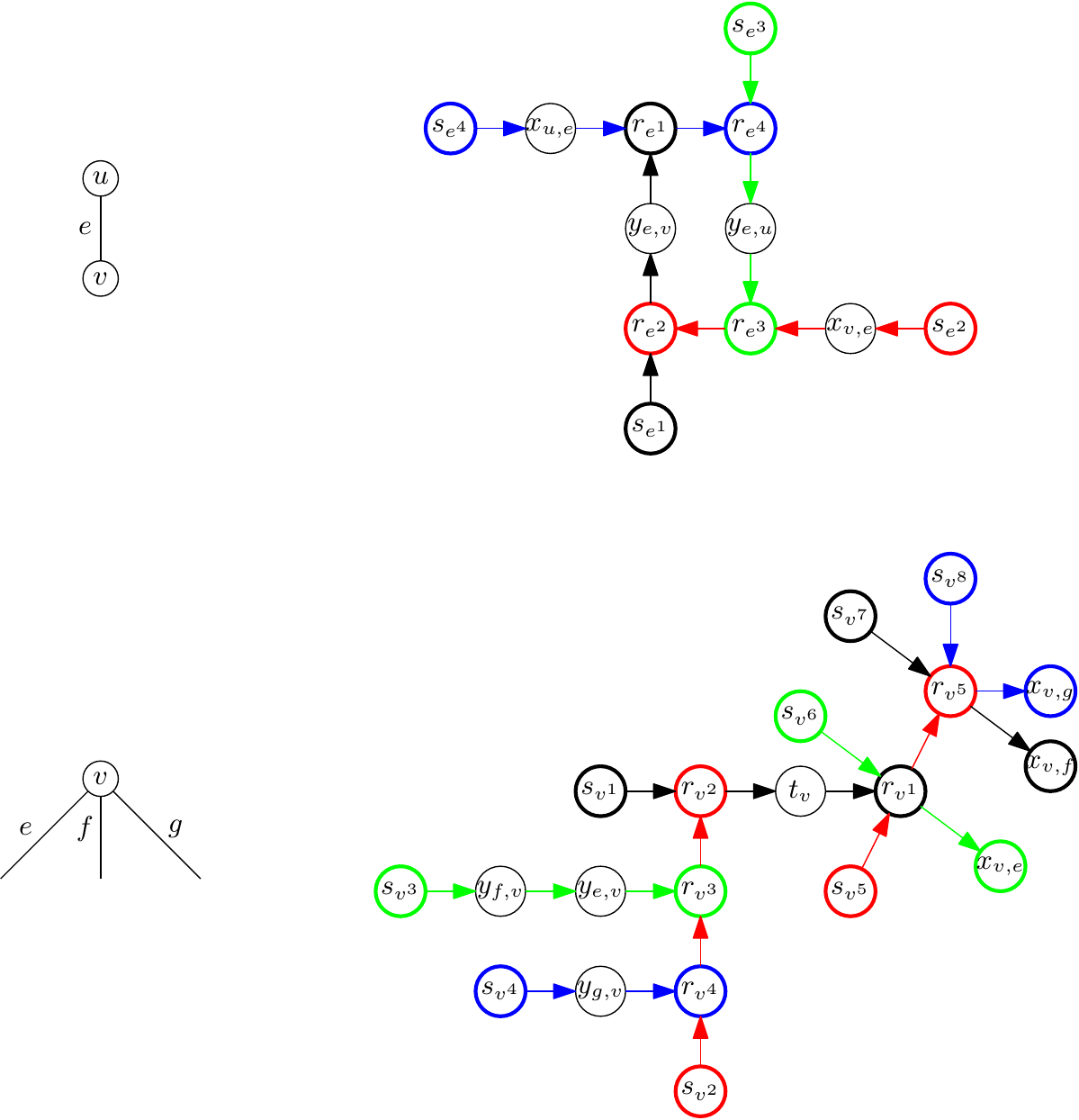}}
\caption{Reduction in the NP-completeness proof (\cref{thm:NPc}). The vertices with the same labels should be identified.}
\label{fig:npc1.pdf}
\end{figure*}

For each edge $e= \{u,v\} \in E$, the agents $e^1,e^2,e^3,e^4$
have the following preferences:
\begin{align*}
    \succ_{e^1} &\colon r_{e^1}, y_{e, v}, r_{e^2}, s_{e^1},&
    \succ_{e^2} &\colon r_{e^2}, r_{e^3}, x_{v, e}, s_{e^2},\\
    \succ_{e^3} &\colon r_{e^3}, y_{e, u}, r_{e^4}, s_{e^3},&
    \succ_{e^4} &\colon r_{e^4}, r_{e^1}, x_{u, e}, s_{e^4}.
\end{align*}
For each vertex $v\in V$ with $\delta (v)=\{e, f, g\}$, 
we define the preferences of the associated $8$ agents as follows:
\begin{align*}
    \succ_{v^1} &\colon r_{v^1}, t_{v}, r_{v^2}, s_{v^1},\ \ \ \ 
    \succ_{v^2} \colon r_{v^2}, r_{v^3}, r_{v^4}, s_{v^2},\\
    \succ_{v^3} &\colon r_{v^3}, y_{e, v}, y_{f,v}, s_{v^3},\ \ \ \ 
    \succ_{v^4} \colon r_{v^4}, y_{g, v}, s_{v^4},\\
    \succ_{v^5} &\colon r_{v^5}, r_{v^1}, s_{v^5},\\
    \succ_{v^6} &\colon r_{v^6}, r_{v^1}, s_{v^6},\quad \text{where\ } r_{v^6}=x_{v, e},\\
    \succ_{v^7} &\colon r_{v^7}, r_{v^5}, s_{v^7},\quad \text{where\ } r_{v^7}=x_{v, f},\\
    \succ_{v^8} &\colon r_{v^8}, r_{v^5}, s_{v^8},\quad \text{where\ } r_{v^8}=x_{v, g}.
\end{align*}
The initial matching $\mu$ is defined to be $\mu(i)=s_i$ for each 
agent $i \in N$.
Then by Claim \ref{clm:NPc1} below, a reformist envy-free matching
$\sigma$ with respect to $\mu$ is $\sigma(i)=r_{i}$ for each agent 
$i \in N$.
We observe that each agent $i \in N$ has a set $M_i$ 
of size at most four, and each 
item appears in $M_i$ for at most three agents $i \in N$.

\begin{claim}\label{clm:NPc1}
If $G$ has a vertex cover of size $k$, 
then there exists a reformist sequence 
of length $|N| + |E| + k$.
\end{claim}
  \begin{proof}
    Let $S$ be a vertex cover of size $k$ in $G$.
    Consider the following reformist sequence.
    \begin{enumerate}
      \item For each vertex $v\in S$, the agents $v^1,v^2,v^3, v^4$ 
      are nominated one by one as follows.
      The agent $v^1$ exchanges $s_{v^1}$ with $t_{v}$.
      Then $v^2$ exchanges $s_{v^2}$ with $r_{v^2}$, and $v^3$ and $v^4$ exchange $s_{v^3}$ and $s_{v^4}$ with $r_{v^3}$ and $r_{v^4}$, respectively.
      This takes $4$ steps for each vertex $v\in S$.
      \item For each edge $e = \{u,v\}\in E$, 
      the agents $e^1,e^2,e^3,e^4$ are nominated one by one.
      Since $S$ is a vertex cover, $u$ or $v$ belongs to $S$.
      By symmetry, suppose that $v \in S$.
      The agent $e^1$ exchanges $s_{e^1}$ with $y_{e, v}$, which can be done because $y_{e, v}$ has no envy from $v^3$ or $v^4$ due to Step 1.
      Then the agent $e^\ell$ exchanges $s_{e^\ell}$ with $r_{e^\ell}$ in the order of $\ell=2,3,4$.
      Finally, the agent $e^1$ exchanges $y_{e, v}$ with $r_{e^1}$.
      This takes $5$ steps for each edge $e\in E$.
      \item For each vertex $v\in V$ and each integer $\ell \in \{6,7,8\}$, 
      $v^\ell$ exchanges $s_{v^\ell}$ with $r_{v^\ell}$, and then $v^5$ exchanges $s_{v^5}$ with $r_{v^5}$.
      This can be done since $r_{v^\ell}$ has no envy from the other agents for each integer $\ell \in \{6,7,8\}$ due to Step 2.
      This takes $4$ steps for each vertex $v\in V$.
      \item For each vertex $v\in S$, $v^1$ exchanges $t_{v}$ with $r_{v^1}$.
      This takes $1$ step for each vertex $v\in S$.
      \item For each vertex $v\in V\setminus S$, the four agents $v^\ell$ exchange $s_{v^\ell}$ with $r_{v^\ell}$ in the order of $\ell=1,2,3,4$.
      This takes $4$ steps for each vertex $v\in V\setminus S$.
    \end{enumerate}
    The total number of steps in the reformist sequence is
    $4k + 5|E| + 4|V| + k + 4(|V|-k) = 8|V|+5|E|+k$.
    Since $|N|=8|V|+4|E|$, this is equal to $|N| + |E| +k$.
  \end{proof}

  \begin{claim}\label{clm:NPc2}
  If there exists a reformist sequence of length $|N| + |E| +k$, then 
  $G$ has a vertex cover of size $k$.
  \end{claim}
  \begin{proof}
    Consider a reformist sequence with minimum length.
    We first observe the following because of the minimality.
    \begin{itemize}
       \item For each vertex $v\in V$, the agents 
       $v^2,\dots,v^8$
       exchange $s_{v^\ell}$ with $r_{v^\ell}$ in the reformist
       sequence, since moving to an intermediate item is redundant, 
       i.e., moving to an intermediate item does not improve the 
       situation of the other agents.
       We note that the agent $v^1$ may use $t_{v}$.
       Thus, for each vertex $v\in V$, we spend eight or nine steps.
       \item For each edge $e=\{u, v\}\in E$, the agent $e^2$~($e^4$, 
       resp.,) exchanges $s_{e^2}$ with $r_{e^2}$~($s_{e^4}$ with 
       $r_{e^4}$, resp.,) in the sequence.
       Moreover, only one of $y_{e, u}$ or $y_{e, v}$ must be used in 
       the sequence.
       Thus, for each edge $e\in E$, we spend exactly $5$ steps.
    \end{itemize}
    Define $S$ as the set of $v\in V$ such that 
    the agent $v^1$ possesses $t_{v}$ at some point.
    Then the number of steps is $8|V|+|S|+5|E|$.
    By the assumption with $|N|=8|V|+4|E|$, 
    it follows that $|S|\leq k$.
    
    We will claim that $S$ is a vertex cover of $G$.
    Indeed, suppose to the contrary that $S$ is not a vertex cover.
    Then there is some edge $e=\{u, v\}\in E$ such that $u\not\in S$ and $v\not\in S$.
    That is, neither $t_{u}$ nor $t_{v}$ is used in the sequence.
    This means that, in the sequence, $y_{e, u}$ and $y_{e, v}$ always have envy from one of $u^\ell$'s and $v^\ell$'s, respectively, and hence neither the agents $e^1$ nor $e^3$ can exchange items, which is a contradiction.
  \end{proof}
  By Claims~\ref{clm:NPc1} and~\ref{clm:NPc2}, the vertex cover problem in $3$-regular graphs is reduced to the decision version of the 
shortest reformist sequence problem, which completes the proof.
\end{proof}

\section{Shortest Reformist Sequence: Algorithms}

\subsection{Preprocessing}
\label{sec:preprocessing}

Here we present some basic observations for the shortest reformist 
sequence problem.
Suppose that $\mu$ is an initial envy-free matching.
Then, as mentioned in Section~\ref{sec:uniqueness}, the reformist matching with respect to $\mu$ can be found in polynomial time, which is denoted by $\sigma$.  
If $\mu(i)\succ_i x$ for some $i\in N$ and $x \in M_i$, then we can 
remove $x$ from $M_i$ because $i$ never 
envies an agent having 
$x$. 
If $x \succ_i \sigma(i)$ for some $i\in N$ and $x \in M_i$, 
then we can remove $x$ from the instance because 
$i$ always envies an agent having $x$. 
Hence, we may assume that $\sigma(i) \succeq_i x \succeq_i \mu(i)$ for 
every item $x \in M_i$. 
We may also assume that $\sigma (i)\succ_i \mu (i)$ for every agent 
$i\in N$, as we can simply remove agents $i$ with 
$\sigma (i) = \mu (i)$.
This implies that the length of every reformist sequence must be at 
least $n$.

We denote 
$S = \{ \mu (i)\mid i\in N\}$ and $R = \{ \sigma (i)\mid i\in N\}$.
Then 
$S\cap R=\emptyset$ holds.
In fact, suppose that there exist two agents $i, j$ such that 
$\mu (i)=\sigma (j)$.
Then, $\mu(i) \succ_j \mu(j)$ since $\sigma(j) \succeq_j \mu(j)$ and
$\mu(i) \neq \mu(j)$. However, this contradicts the envy-freeness of
$\mu$.

Those assumptions can be ensured in polynomial time, and thus
in the sequel we assume that given instances satisfy those properties.

\subsection{Preferences of length three}

While \cref{thm:NPc} says 
the shortest reformist sequence problem is NP-hard
when each agent has at most four acceptable items,
we show that, if each agent has at most three acceptable items, then the shortest reformist sequence problem is polynomial-time solvable.

\begin{theorem}
\label{thm:atmost3items}
If $m_i \le 3$ for every agent $i \in N$,
then a shortest reformist sequence can be found in polynomial time.
\end{theorem}
\begin{proof}
We prove that there is 
a reformist sequence of length $n=|N|$, and such a sequence can be found in polynomial time.
Since $n$ is a lower bound on the length of a reformist sequence~(see Section~\ref{sec:preprocessing}), the obtained sequence is optimal.
Let $\mu$ and $\sigma$ be an 
initial envy-free matching and 
the reformist envy-free matching with respect to $\mu$, respectively. 

As mentioned in Section~\ref{sec:preprocessing}, we may assume that 
$\{\sigma(i)\mid i\in N\}$ and $\{\mu(i)\mid i\in N\}$ are disjoint.
Thus, we may assume that each agent $i\in N$ has preferences 
$\sigma (i) \succ_i b(i) \succ_i \mu(i)$ by appending a dummy 
item $b(i)$ if $m_i<3$.

We claim that there exists an agent $i \in N$ 
such that $i$ can exchange $\mu (i)$ with $\sigma (i)$ keeping 
envy-freeness.
If this claim is true, by repeatedly finding such an agent and 
removing her, we can find a reformist sequence of length $n$, 
which completes the proof.

To find such an agent, we construct an auxiliary directed graph 
$D = (N, A)$ where, for each pair $i,j \in N$, an arc $(i,j) \in A$ 
exists if and only if $\sigma (i) = b(j)$
(see \figurename~\ref{fig:algo_lengththree1}). 
The existence of the arc $(i,j)$ means that $i$ 
cannot exchange to $\sigma (i)$ before $j$ gets $\sigma (j)$. 
Since $\sigma$ is a reformist envy-free matching with respect 
to $\mu$, 
there does not exist a directed cycle in $D$, that is, $D$ is acyclic. 
Hence $D$ has a sink $i\in N$. 
Since $\sigma (i)\neq b(j)$ for every agent $j\in N$, 
the agent $i$ can exchange $\mu (i)$ with $\sigma (i)$
(see \figurename~\ref{fig:algo_lengththree2}).
Thus the claim follows.
\end{proof}

\begin{figure*}[ht]
    \centering
    \resizebox{0.9\textwidth}{!}{\includegraphics{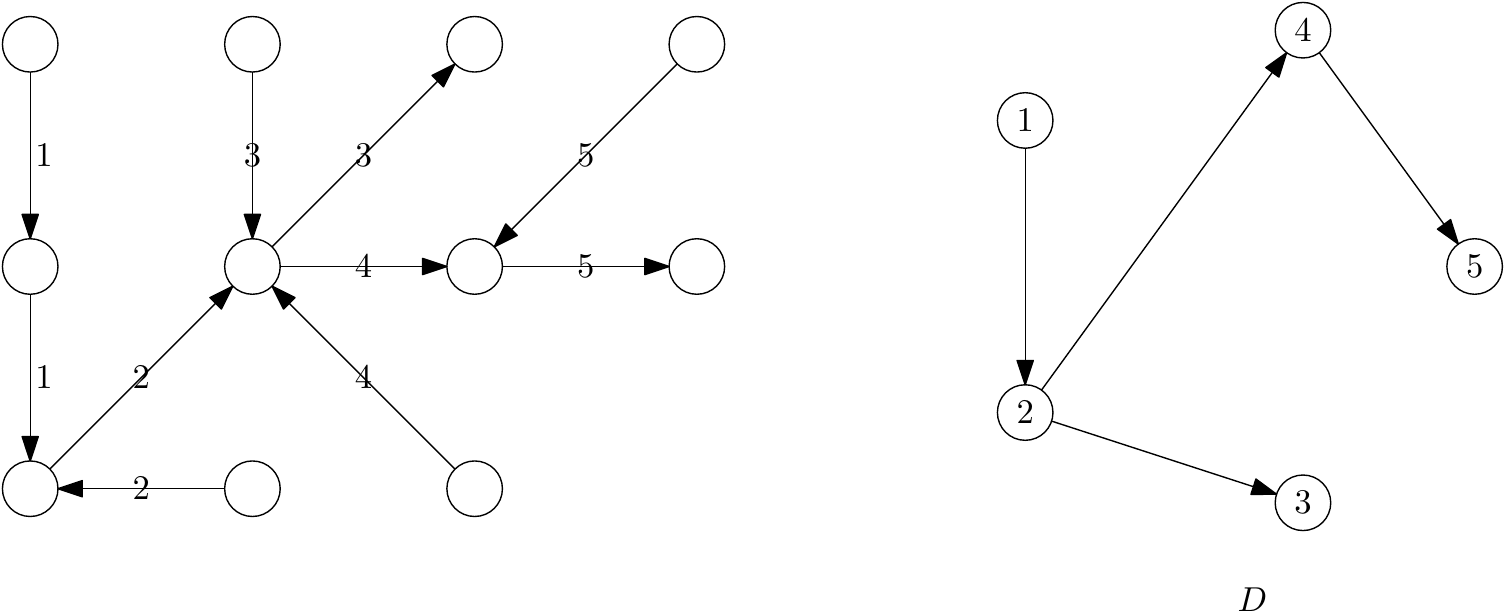}}
    \caption{Algorithm for \cref{thm:atmost3items}. (Left) An instance. (Right) The construction of the directed graph $D$.}
    \label{fig:algo_lengththree1}
\end{figure*}

\begin{figure*}[ht]
    \centering
    \resizebox{0.9\textwidth}{!}{\includegraphics{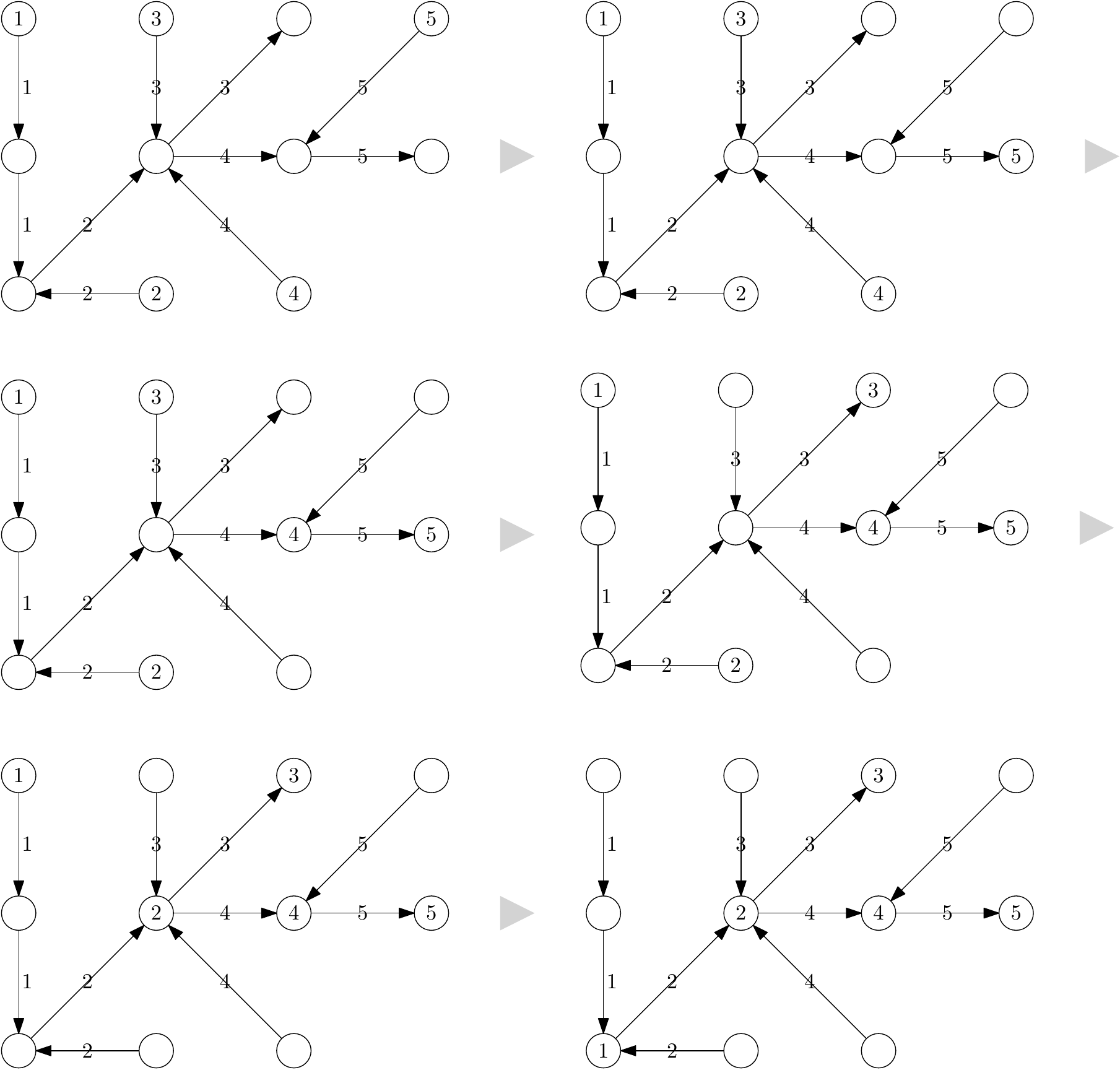}}
    \caption{Algorithm for \cref{thm:atmost3items}.
    A shortest reformist sequence following the order $5, 4, 3, 2, 1$.}
    \label{fig:algo_lengththree2}
\end{figure*}

\subsection{Items are acceptable to at most two agents}

By \cref{thm:NPc}, 
the shortest reformist sequence problem is NP-hard
when each item is acceptable to at most three agents.
Here,
we show that, if every item is acceptable to at most two agents,
then the shortest reformist sequence problem is polynomial-time solvable.

\begin{theorem}\label{thm:<=2}
If $|\{ i \in N \mid x \in M_i \}| \leq 2$ for every item $x \in M$,
then we can obtain a shortest reformist sequence in polynomial time.
\end{theorem}

To this end,
we introduce a slightly generalized version of the original problem.
Let $\mu$ be an initial envy-free matching and $\sigma$ the reformist envy-free matching with respect to $\mu$.
For each agent $i\in N$, we are given an item $x_i \in M_i$, and define
$L_i := \{x \in M_i \mid x \succeq_i x_i\}$
as the {\it target set} of $i$.
Note that $L_i$ consists of 
the best $|L_i|$ items with respect to $\succ_i$. 
Denote $\mathcal{L} := \{L_i \mid i \in N\}$.
In addition,
we are given a partition $\mathcal{N} := \{N_1,N_2,\dots,N_k\}$ 
of $N$.

We generalize the concept of envy according to the target sets 
$\mathcal{L}$ and the partition $\mathcal{N}$ as follows.
Suppose that an agent $i$ is in $N_a$ and an agent $j$ is in 
$N_b$.
For a matching $\mu'$, we say that $i$ has 
\emph{$(\mathcal{L}, \mathcal{N})$-envy} for $j$ on $\mu'$ 
if
\begin{equation*}
\begin{cases}
\mu'(j) \succ_i \mu'(i) & \mbox{if $a = b$}, \\
\mu'(j) \succ_i \mu'(i) \mbox{ and } 
\mbox{$\mu'(i^{\prime}) \notin L_{i^{\prime}}$ ($\forall i^{\prime} \in N_a$)} 
& \mbox{if $a \neq b$}.
\end{cases}
\end{equation*}
The definition says that an agent $i$ has envy for $j$ if the 
agent $i$ prefers $\mu' (j)$ to $\mu' (i)$, except in the 
case when $i$ and $j$ are in different groups and some agent 
$i'$ in the same group as $i$ has an item in $L_i$.
In other words, if some agent $i^\prime$ is assigned an item in 
her target set $L_{i^\prime}$, 
then agents in the same group as $i^\prime$ have no envy for agents in the other groups.
A matching $\mu'$ is said to be \emph{$(\mathcal{L}, \mathcal{N})$-envy-free} if 
every agent $i \in N$ has no $(\mathcal{L}, \mathcal{N})$-envy for any agent $j \in N \setminus \{i\}$ on $\mu'$.
An $(\mathcal{L}, \mathcal{N})$-envy-free matching $\mu'$ is
said to be \emph{satisfactory}
if, for each index $a \in \{1,2,\dots,k\}$,
there is an agent $i \in N_a$ such that $\mu'(i) \in L_i$.
Since an envy-free matching is $(\mathcal{L}, \mathcal{N})$-envy-free and the reformist envy-free matching is satisfactory,
there always exists a sequence $\mu \leadsto \dots \leadsto \sigma'$ of $(\mathcal{L}, \mathcal{N})$-envy-free matchings
from $\mu$ to some satisfactory matching $\sigma'$.
In what follows, we consider the problem of
finding such a sequence with minimum length.

The following theorem shows that the problem defined above can be solved in polynomial time
if every item is acceptable by at most two agents.
\begin{theorem}\label{thm:<=2:g}
If $|\{ i \in N \mid x \in M_i \}| \leq 2$ for every item $x \in M$,
then a shortest sequence of $(\mathcal{L}, \mathcal{N})$-envy-free matchings to some satisfactory matching can be found in polynomial time.
\end{theorem}
We remark that, if $|N_a| = 1$ for every 
$a \in \{1,2,\dots, k\}$ and 
$|L_i|=1$ for every agent $i \in N$,
the above problem is equivalent to the shortest reformist sequence problem.
Thus, \cref{thm:<=2} immediately follows from \cref{thm:<=2:g}.

\begin{proof}[Proof of \cref{thm:<=2:g}]
We denote by $\ell\left(\mathcal{L}, \mathcal{N}\right)$ 
the length of a shortest sequence from an initial matching $\mu$ to some satisfactory matching.
If the value
$\sum_{i \in N} |M_i \setminus L_i|$
is zero,
then $\mu$ is satisfactory, which implies $\ell\left(\mathcal{L}, \mathcal{N}\right) = 0$.

We assume $\sum_{i \in N} |M_i \setminus L_i| > 0$.
In the following,
we construct in polynomial time a new instance $(\mathcal{L}^{\ast}, \mathcal{N}^{\ast})$ with the set $N^{\ast}$ of agents, 
where $\mathcal{N}^{\ast}$ is a partition of $N^{\ast}$ and $\mathcal{L}^{\ast} := \{L_i^{\ast} \mid i \in N^{\ast}\}$ is the set of the target sets $L_i^{\ast}$,  
satisfying the following two conditions.
\begin{description}
\item[C1.]
$\sum_{i \in N^{\ast}} |M_i \setminus L_i^{\ast}| + |N^{\ast}| < \sum_{i \in N} |M_i \setminus L_i| + |N|$.
\item[C2.]
$\ell\left(\mathcal{L}^{\ast}, \mathcal{N}^{\ast}\right) = \ell\left(\mathcal{L}, \mathcal{N}\right) - c$,
where the value $c$ can be computed in polynomial time from the original instance $\left(\mathcal{L}, \mathcal{N}\right)$.
\end{description}
If such an instance can be constructed, then 
we can obtain $\ell\left(\mathcal{L}, \mathcal{N}\right)$
in polynomial time by recursive computation:
the condition~C1 implies that the number of recursive calls is bounded by $\sum_{i \in N} |M_i| + |N|$;
the condition~C2 verifies that each recursive call can be preformed in polynomial time
and that we can compute $\ell\left(\mathcal{L}, \mathcal{N}\right)$
from $\ell\left(\mathcal{L}^{\ast}, \mathcal{N}^{\ast}\right)$.

We first consider a simple case where there exist an index $a \in \{1,2,\dots, k\}$ and
an agent $i \in N_a$ such that $\mu(i) \in L_i$.
Then we define $N^{\ast} := N \setminus N_a$, $\mathcal{N}^{\ast} := \mathcal{N} \setminus \{N_a\}$, and $L_i^{\ast} := L_i$ for each agent $i\in N^{\ast}$.
Moreover, we set an initial envy-free matching as the restriction of $\mu$ to $N^{\ast}$.
Then the new instance satisfies the conditions~C1 and C2 since $|N^{\ast}|<|N|$ and $\ell\left(\mathcal{L}^{\ast}, \mathcal{N}^{\ast}\right) =\ell\left(\mathcal{L}, \mathcal{N}\right)$.

A similar argument can be applied to the case when there exist an index $a \in \{1,2,\dots, k\}$ and
an agent $i \in N_a$ such that $\mu(i) \notin L_i$ but
a matching $\mu'$ obtained from $\mu$ by exchanging $\mu (i)$ with 
some $x \in L_i$
is $(\mathcal{L}, \mathcal{N})$-envy-free.
In this case, we define $N^{\ast} := N \setminus N_a$, $\mathcal{N}^{\ast} := \mathcal{N} \setminus \{N_a\}$, and $L_i^{\ast} := L_i$ for each agent $i\in N^{\ast}$.
We set an initial envy-free matching as the restriction of $\mu'$ to $N^{\ast}$.
Then the new instance satisfies the conditions~C1 and C2 since $|N^{\ast}|<|N|$ and $\ell\left(\mathcal{L}^{\ast}, \mathcal{N}^{\ast}\right) =\ell\left(\mathcal{L}, \mathcal{N}\right)-1$.

Therefore, we may assume the following.

\smallskip
\noindent
($\ast$) For any agent $i$,  $\mu(i) \not\in L_i$ and any item $x\in L_i$ receives $(\mathcal{L}, \mathcal{N})$-envy from some agent on $\mu(i)$.
\smallskip
Since each item is acceptable to at most two agents, there exists exactly one agent that has $(\mathcal{L}, \mathcal{N})$-envy for an item $x\in L_i$.

We construct a new instance $(\mathcal{L}^{\ast}, \mathcal{N}^{\ast})$ as follows.
Let $N^{\ast} := N$.
We define the %
directed graph $D = (V,A)$ by 
\begin{align*}
\begin{split}
V & := \{1,2,\dots,k\}, \\
A & := \{(a,b) \in V \times V \mid a \neq b, 
\exists i \in N_a, \exists j \in N_b,\ M_i \cap L_j \neq \emptyset\}.  
\end{split}
\end{align*}
Roughly, $M_i \cap L_j \neq \emptyset$ means that $j$ cannot receive some item in $L_j$ because of $(\mathcal{L}, \mathcal{N})$-envy from $i$. 
We decompose $D$ into strongly connected components. 
Let $S \subseteq V$ be a source component of the decomposition (i.e., 
no arc in $A$ enters $S$). 
We define the partition $\mathcal{N}^{\ast}$ of $N^{\ast}$ by merging 
$\{N_a \mid a \in S \}$
into one part,
i.e.,
    $\mathcal{N}^{\ast} := \big(\mathcal{N} \setminus \{ N_a \mid a \in S \}\big) \cup \left\{N_S\right\}$,
where we define $N_S := \bigcup_{a \in S} N_a$.
Let
$X := \bigcup_{i \in N_S} L_i$. 
For each agent $i \in N_S$, we denote by $x_i^{\ast}$ be the item in 
$X \cap M_i$ that is minimum with respect to $\succ_i$. 
We define $\mathcal{L}^{\ast}=\{L^{\ast}_i\mid i\in N^{\ast}\}$ by 
\begin{equation*}
L_i^{\ast} := 
\begin{cases}
\{x \in M_i \mid x \succeq_i x_i^{\ast}\} & \text{ if } i\in N_S, \\
L_i &\text{ if } i \in N\setminus N_S. 
\end{cases}
\end{equation*}
We set $\mu$ as the initial matching in the resulting instance.
Note that $\mu$ is an $(\mathcal{L}^{\ast}, \mathcal{N}^{\ast})$-envy-free matching.

\cref{cl:<,cl:-|S|} below show that the resulting instance satisfies the conditions~C1 and~C2, respectively.
\begin{claim}\label{cl:<}
$\sum_{i \in N} |M_i \setminus L_i^{\ast}| < \sum_{i \in N} |M_i \setminus L_i|$.
\end{claim}
\begin{proof} 
By definition, $L_i^{\ast} \supseteq L_i$ holds for every agent $i \in N_S$.
Hence, it suffices to show $L_i^{\ast} \supsetneq L_i$ for some agent 
$i \in N_S$.

Suppose, to the contrary,
that $L_i^{\ast} = L_i$ for every agent $i \in N_S$.
Take an arbitrary sequence of $(\mathcal{L}, \mathcal{N})$-envy-free matchings from $\mu$ to some satisfactory matching.
Let $\mu'$ denote the first $(\mathcal{L}, \mathcal{N})$-envy-free matching in the sequence
with $\mu'(i) \in L_i$ for some agent $i \in N_S$.
Recall here that
$\mu'(i)$ belongs to the sets of acceptable items of two agents;
one is $i$
and the other is denoted by $j$.
Since $S$ is a source component of $D$,
the group $N_b$ having $j$ belongs to $N_S$.
Since $\mu'(i)\in X$, we see that $\mu'(i)\in L_j^{\ast}$.
Since  $L_j^{\ast} = L_j$ by assumption,
$\mu'(i)$ particularly belongs to $L_j$,
which implies that $j$ has $(\mathcal{L}, \mathcal{N})$-envy for $i$.
This is a contradiction to the $(\mathcal{L}, \mathcal{N})$-envy-freeness of $\mu'$.
\end{proof}

\begin{claim}\label{cl:-|S|}
$\ell\left(\mathcal{L}^{\ast}, \mathcal{N}^{\ast}\right) = \ell\left(\mathcal{L}, \mathcal{N}\right) - |S|$.
\end{claim}
\begin{proof} %
Let $\ell := \ell\left(\mathcal{L}, \mathcal{N}\right)$ and 
$\ell^{\ast} := \ell\left(\mathcal{L}^{\ast}, \mathcal{N}^{\ast}\right)$.

We first show $\ell^{\ast} \leq \ell - |S|$.
Consider a shortest sequence 
$\mu =: \mu_0 \leadsto \mu_1 \leadsto \dots \leadsto \mu_\ell$ of $(\mathcal{L}, \mathcal{N})$-envy-free matchings from 
$\mu$ to some satisfactory $(\mathcal{L}, \mathcal{N})$-envy-free
matching $\mu_\ell$.
For each $a \in S$,
we denote by $p_a$ the minimum index
such that $\mu_{p_a}(i) \in L_i$ for some agent $i \in N_a$.
Let $b\in S$ be the index that satisfies $p_b = \min\{p_a \mid a \in S\}$, and $i_b\in N_b$ be the agent such that $\mu_{p_b}(i_b) \in L_{i_b}$.
Then, by assumption~{\rm ($\ast$)}, the item $\mu_{p_b}(i_b)$ is acceptable to another agent $j$ in some group $N_a$, meaning that either $a=b$ or $D$ has an arc $(a, b)$.
Since $S$ is a source component, we see $a\in S$.
By the definition of $p_b$, we observe that $\mu_{p_b -1}(j)\succ_j \mu_{p_b}(i_b)$, which implies that $\mu_{p_b -1}(j)\in L_j^{\ast} \setminus L_j$ as $\mu_{p_b}(i_b)\in X$.
Thus, in the $(p_b-1)$-st step, the agent $j\in N_S$ has an item in the new target set $L_j^{\ast}$.

We construct a sequence of 
$(\mathcal{L}^{\ast}, \mathcal{N}^{\ast})$-envy-free 
matchings as follows.
For $p$ with $p_b \leq p \leq \ell$,
define a matching $\mu_{p}'$ to be $\mu_{p}'(i)= \mu_{p_b-1}(i)$ if $i\in N_S$ and $\mu_{p}'(i)= \mu_{p}(i)$ otherwise.
Then 
the sequence $(\mu_0, \mu_1, \dots, \mu_{p_b-1}, \mu'_{p_b}, \dots, \mu'_\ell)$
forms that of $(\mathcal{L}^{\ast}, \mathcal{N}^{\ast})$-envy-free matchings
from $\mu$ to some satisfactory $(\mathcal{L}^{\ast}, \mathcal{N}^{\ast})$-envy-free matching.
Furthermore, since $\mu'_{p_a} = \mu'_{p_a-1}$ holds 
for all $a\in S$,  we can remove $\mu'_{p_a}$'s from the sequence.
This implies that there exists a sequence of $(\mathcal{L}^{\ast}, \mathcal{N}^{\ast})$-envy-free matchings whose length is $\ell - |S|$.
Thus $\ell^{\ast} \leq \ell - |S|$ holds.

We next show $\ell^{\ast} \geq \ell - |S|$.
Consider a shortest sequence $\mu =: \mu_0 \leadsto \mu_1 \leadsto \dots \leadsto \mu_{\ell^{\ast}}$ of $(\mathcal{L}^{\ast}, \mathcal{N}^{\ast})$-envy-free matchings from $\mu$ to some satisfactory $(\mathcal{L}^{\ast}, \mathcal{N}^{\ast})$-envy-free matching $\mu_{\ell^{\ast}}$.
Let $p$ be the minimum index such that $\mu_{p}(i_0) \in L_{i_0}^{\ast}$ for some $a_0 \in S$ and $i_0 \in N_{a_0}$.
We observe that $\mu_{p}(i_0)$ particularly belongs to $L_{i_0}^{\ast} \setminus L_{i_0}$.
In fact, suppose that $\mu_{p}(i_0)\in L_{i_0}$.
Then the item $\mu_{p}(i_0)$ is acceptable to another agent by assumption~{\rm ($\ast$)}.
Since $S$ is a source component, there exist an index $b\in S$ and $j\in N_b$ such that $\mu_{p}(i_0)\in M_j$.
Since $\mu_{p}(i_0)\in L^{\ast}_j$ and $\mu_{p-1}(j)\succ_j \mu_{p}(i_0)$, we have $\mu_{p-1}(j)\in L^{\ast}_j$.
This contradicts the definition of $p$.

Let $x_1 := x^{\ast}_{i_0}$.
Then there exist $a_1 \in S$ and $i_1 \in N_{a_1}$ such that $x_1 \in L_{i_1}$.
Moreover, since $\mu_{p}(i_0) \neq x_1$ as $\mu_{p}$ is $(\mathcal{L}^{\ast}, \mathcal{N}^{\ast})$-envy-free, it follows that $\mu_{p}(i_0) \succ_{i_0} x_1$.
Hence, the matching $\sigma_1$ obtained from $\mu_p$ by assigning $x_1$ to $i_1$
is $(\mathcal{L}, \mathcal{N})$-envy-free.
By this change, $N_{a_1}$ has the agent $i_1$ with an item in her target set, and hence any agent in $N_{a_1}$ has no $(\mathcal{L}, \mathcal{N})$-envy for agents in the other groups on $\sigma_1$.

Suppose that there exists some $a_2 \in S$ with $(a_1, a_2) \in A$.
Then there exist $i'\in N_{a_1}$ and $i_2\in N_{a_2}$ such that $M_{i'}\cap L_{i_2}$ contains an item, say $x_2$. 
Since the agent $i'$ has no $(\mathcal{L}, \mathcal{N})$-envy for agents in $N_{a_2}$ on $\sigma_1$, 
the matching $\sigma_2$ obtained from $\sigma_1$ by assigning $x_2$ to $i_2$
is $(\mathcal{L}, \mathcal{N})$-envy-free.
This change makes $N_{a_2}$ have the agent $i_2$ with an item in her target set.
We repeat the above procedure; 
for $j=2,3,\dots$, we find an index $a_j\in S \setminus \{a_1, \dots , a_{j-1} \}$ such that there exists an arc to $a_j$ from 
$\{a_1, \dots , a_{j-1} \}$, 
and obtain an $(\mathcal{L}, \mathcal{N})$-envy-free matching by exchanging an item for some agent in $N_{a_j}$. 
Since $D[S]$ forms a strongly connected component of $D$,
the repetition can be performed until, for all $a \in S$, some agent in $N_{a}$ has an item in her target set.
Thus we obtain a sequence $\mu_p \leadsto \sigma_1 \leadsto \sigma_2 \leadsto \dots \leadsto \sigma_{|S|}$ of $(\mathcal{L}, \mathcal{N})$-envy-free matchings
in which for each $a \in S$ there is an agent $i \in N_a$ with $\sigma_{|S|}(i) \in L_i$.

For $q$ with $p+1 \leq q \leq \ell^{\ast}$,
we define a matching $\mu_{q}'$ to be $\mu_{q}'(i)=\sigma_{|S|}(i)$ if $i \in N_S$ and $\mu_{q}'(i)=\mu_{q}(i)$ otherwise.
Then $\mu_0 \leadsto \dots \leadsto \mu_p \leadsto \sigma_1 \leadsto \dots \leadsto \sigma_{|S|} \leadsto \mu_{p+1}' \leadsto \dots \leadsto \mu_{\ell^{\ast}}'$
forms a sequence of $(\mathcal{L}, \mathcal{N})$-envy-free matchings
from $\mu$ to a satisfactory $(\mathcal{L}, \mathcal{N})$-envy-free matching $\mu_{\ell^{\ast}}'$.
This implies $\ell \leq \ell^{\ast} + |S|$.
\end{proof}

It follows from the above claims that we can recursively compute $\ell (\mathcal{L}, \mathcal{N})$ in polynomial time.
Since the above proofs are constructive,
we can find a shortest sequence as well in polynomial time.
This completes the proof. %
\end{proof}

\section{Coping with NP-Hardness}

To cope with NP-hardness of the shortest reformist sequence
problem,
we need to compromise either obtaining exact solutions or
computing in polynomial time.

The compromise of exact solutions leads us to the possibility of
approximation algorithms.
A polynomial-time algorithm for the shortest reformist sequence problem
\emph{approximates within a factor of $\alpha\geq 1$} if
for all instances, 
the length of the reformist sequence that is obtained as the output
of the algorithm is at most as long as $\alpha$ times the shortest 
length of a reformist sequence.
The smaller value of $\alpha$ means a better approximation guarantee.

It turns out that even an approximation is hard to obtain.
The following theorem gives a precise statement of the sentence above.

\begin{theorem}
\label{thm:inapx}
The shortest reformist sequence problem is inapproximable 
in polynomial time 
within a factor of $c \ln n$ for some constant $c>0$, unless \textup{P} $=$ \textup{NP}.
\end{theorem}
\begin{proof}
    We reduce the set cover problem to the shortest reformist sequence problem.
    In the set cover problem, we are given a family of subsets $\mathcal{S}=\{S_1, \dots, S_h\}$ on the ground set $V$ where we assume that $h=O({|V|}^\beta)$.
    The goal is to find a subfamily $\mathcal{S}'\subseteq \mathcal{S}$ such that $\bigcup_{S\in\mathcal{S}'}S=V$ and $|\mathcal{S}'|$ is minimized.
    It is known~\cite{DS14} that the set cover problem is inapproximable within a factor of $(1-\varepsilon)\ln n$ for every $\varepsilon >0$ unless P$=$NP.
    We denote by $\delta (v)$ the set of subsets in $\mathcal{S}$ that contain an element $v\in V$, i.e., $\delta (v) = \{S\in \mathcal{S} \mid v\in S\}$.
    Also, we define $T := \sum_{S\in\mathcal{S}}|S|$.

    We construct an instance of the shortest reformist sequence problem as follows.
    Let $p$ be a positive integer, which will be specified later.
    For each subset $S_j\in \mathcal{S}$ where $d_j=|S_j|$, we define $d_j+3$ agents $x_{j,0}, x_{j,1},\dots, x_{j, d_j}$, $y_j$, and $y'_j$, and, for each element $v\in V$, we define $f_v$ agents 
    $v^1,v^2,\dots,v^{f_v}$, where $f_v=|\delta (v)|$.
    Moreover, we introduce one more agent $z$.
    Thus 
    \begin{equation*}
    \begin{split}
    N & = 
    \{ x_{j,\ell} \mid j \in \{1,\dots, h\}, \ell\in \{0,1,2,\dots, d_j\}\}
    \cup \{ y_j, y'_j \mid j \in \{1,\dots, h\}\} \\
    & \ \ \ \ 
    \cup \{ v^\ell \mid v\in V, \ell \in \{ 1,2,\dots, f_v\} \}
    \cup\{z\}.
    \end{split}
    \end{equation*}
    We see that $n=|N|=O(h|V|)$.

    Define 
    \begin{equation*}
    M'=\{r_i, s_i\mid i\in N\}\cup \{t_{j, v}\mid S_j\in \mathcal{S}, v\in S_j\}
    \cup \{u_j\mid S_j\in\mathcal{S}\}.
    \end{equation*}
    For each subset $S_j\in \mathcal{S}$, we prepare the set of items 
    \begin{equation*}
    M_j = \{a_{j, \ell}, b_{j, \ell}\mid S_j\in\mathcal{S}, \ell \in \{1,2,\dots, p\} \}.
    \end{equation*}
    Our instance has the set of items $M=M'\cup \bigcup_{S_j\in\mathcal{S}} M_j$.

    For each subset $S_j\in \mathcal{S}$, 
    letting $S_j=\{v_1, \dots, v_{d_j}\}$ (in an arbitrary order), 
    the agents $x_{j, \ell}$'s have the following preferences:
     \begin{align*}
       \succ_{x_{j, 0}} &\colon r_{x_{j, 0}}, u_{j}, r_{x_{j,d_j}}, r_{x_{j,d_j-1}}, \dots, r_{x_{j, 1}}, s_{x_{j, 0}},\\
       \succ_{x_{j, \ell}} &\colon r_{x_{j, \ell}}, t_{j, v_\ell}, s_{x_{j, \ell}}\quad (\ell \in \{1,2,\dots, d_j\}).
     \end{align*}
    Moreover, the agents $y_j$ and $y'_j$ have the preferences in a way similar to Theorem~\ref{thm:exponential-example}, that is,
    \begin{align*}
        \succ_{y_j} &\colon a_{j, p}, b_{j, p}, a_{j, p-1}, b_{j, p-1}, \dots, a_{j, 2}, b_{j, 2}, a_{j, 1},\\
        \succ_{y'_j} &\colon b_{j, p}, u_{j}, b_{j, p-1}, a_{j, p}, b_{j, p-2}, a_{j, p-1}, \dots, b_{j, 2}, a_{j, 3}, b_{j, 1}.
    \end{align*}
    For each vertex $v\in V$, we may re-index subsets in 
    $\delta(v)$ so that $\delta (v)=\{S_1, \dots, S_{f_v}\}$.
    Then the preferences of the associated agents are defined to be
    \begin{align*}
    \succ_{v^\ell} &\colon r_{v^\ell}, t_{\ell, v}, r_{v^{\ell+1}}, s_{v^\ell} \quad (\ell \in \{1,2,\dots, f_v\}),
    \end{align*}
    where we assume that $v^{\ell+1} = v^1$.
    Finally, the last agent $z$ has the preference defined by
    \begin{align*}
    \succ_{z} &\colon r_{z}, r_{x_{1, 0}}, \dots, r_{x_{h, 0}}, s_{z}.
    \end{align*}
    The initial matching $\mu$ is set $\mu(i)=s_i$ for each agent $i$.
    By Claim \ref{clm:inapprox1} below, a reformist envy-free matching $\sigma$ with respect to $\mu$ is $\sigma(i)=r_{i}$ for each agent $i$.

    \begin{claim}\label{clm:inapprox1}
    Let $k$ be a non-negative integer.
    If $\mathcal{S}$ has a set cover of size $k$, then there exists a reformist sequence of length at most $(2p-4)k + 2 T  + 4 h + |V| + 1$.
    \end{claim}
    \begin{proof}
        Let $\mathcal{S}^\ast$ be a set cover of size $k$.
        Consider the following reformist sequence.
        \begin{enumerate}
        \item For each $S_j\in \mathcal{S}^\ast$ where $S_j=\{v_1, \dots, v_{d_j}\}$, we do the following:
        The agents $y_j$ and $y'_j$ exchange items repeatedly to obtain $a_{j, p}$ and $b_{j, p}$, respectively, which takes $2p-2$ steps~(See the proof of Theorem~\ref{thm:exponential-example}).
        Then since an envy at $u_{j}$ from $y'_j$ is removed, the agent $x_{j, 0}$ exchanges $s_{x_{j, 0}}$ with $u_{j}$.
        For each integer 
        $\ell \in \{1,2,\dots, d_j\}$, $x_{j, \ell}$ exchanges $s_{x_{j, \ell}}$ with $r_{x_{j,\ell}}$.
        This step takes $2p-2 + d_j +1$ exchanges for each $S_j\in \mathcal{S}^\ast$.
        \item For each element $v\in V$, the associated agents are nominated one by one as follows.
        Since $\mathcal{S}^\ast$ is a set cover, there exists a subset $S\in \mathcal{S}^\ast$ with $v\in S$.
        So we may re-index subsets in $\delta (v)$ so that $\delta (v)=\{S_1, \dots, S_{f_v}\}$ with $S_1\in \mathcal{S}^\ast$.
        The agent $v^1$ exchanges $s_{v^1}$ with $t_{1, v}$.
        This can be done since $t_{1, v}$ receives no envy from $x_{1, \ell}$'s due to Step 1.
        Then the agent $v^\ell$ exchanges $s_{v^\ell}$ with $r_{v^\ell}$ in the order of $\ell=2,3,\dots, f_v$.
        Finally, the agent $v^1$ exchanges $t_{1, v}$ with $r_{v^1}$.
        In total, this step takes $f_v+1$ exchanges for each element $v\in V$.
        \item The agent $z$ exchanges $s_z$ with $r_z$.
        \item For each subset $S_j\in \mathcal{S}^\ast$,
        the agent $x_{j, 0}$ exchanges $u_{j}$ with $r_{x_{j, 0}}$.
        Furthermore, for each subset $S_j\in \mathcal{S}\setminus \mathcal{S}^\ast$,
        the agent $x_{j, \ell}$ exchanges $s_{x_{j, \ell}}$ with $r_{x_{j, \ell}}$ in the order of $\ell=0, 1,\dots, d_j$.
        This can be done since $r_{x_{j, 0}}$ receives no envy from the other agents.
        \item For each subset 
        $S_j\in \mathcal{S}\setminus \mathcal{S}^\ast$,
        $y'_j$ exchanges $s_{y'_j}$ with $u_{j}$, 
        $y_j$ exchanges $s_{y_j}$ with $r_{y_j}$,
        and $y'_j$ exchanges $u_{j}$ with $r_{y'_j}$.
        \end{enumerate}
        In the above reformist sequence, the total number of steps is
        \begin{align*}
         & (2p-2)k + \sum_{S_j\in \mathcal{S}^\ast}(d_j+1) + \sum_{v\in V}(f_v+1)+ 1
          + |\mathcal{S}^\ast| 
           + \sum_{S_j\in \mathcal{S}\setminus \mathcal{S}^\ast}(d_j+1)
          + 3 |\mathcal{S}\setminus \mathcal{S}^\ast|\\
         & = (2p-1)k  + (T + h) 
          + \left(\sum_{v\in V}f_v+|V|\right)+ 1+3(h-k)\\  
         & = (2p-4)k + 2 T  + 4 h + |V| + 1. 
        \end{align*}
        where the last equality holds since $\sum_{v\in V}f_v=T$.
    \end{proof}

    \begin{claim}\label{clm:inapprox2}
    Let $k$ be a non-negative integer.
    If there exists a reformist sequence of length $(2p-4)k + 2 T  + 4 h + |V| + 1$, then $\mathcal{S}$ has a set cover of size $k$.
    \end{claim}
    \begin{proof}
        Consider a reformist sequence of length at most $(2p-4)k + 2 T  + 4 h + |V| + 1$.
        We may assume that it has no redundant steps.
        Define $\mathcal{S}^\prime$ as the family of $S_j\in\mathcal{S}$
        such that $x_{j, 0}$ possesses $u_j$ at some point in the reformist sequence.
        Then, for $S_j\in\mathcal{S}^\prime$, $x_{j, 0}$ takes $2$ steps, and $x_{j, \ell}$ takes one step for each $\ell =1,2,\dots, d_j$. 
        Moreover, for such $j$, $y_j$ and $y'_j$ have to take $2p-2$ steps in total to remove envy at $u_j$ before the agent $x_{j, 0}$ moves.
        For $S_j\in\mathcal{S}\setminus \mathcal{S}^\prime$, the agents $x_{j, \ell}$ takes one step for each $\ell =0, 1,2,\dots, d_j$, and $y_j$ and $y'_j$ take at least $3$ steps.
        For each element $v\in V$, the agents $v^\ell$'s take at least $f_v+1$ steps in total.
        Therefore, the total number of steps in the reformist sequence is at least
        \begin{align*}
         & (2p-2)|\mathcal{S}^\prime| + \sum_{S_j\in \mathcal{S}^\prime}(d_j+2) + \sum_{S_j\in \mathcal{S}\setminus \mathcal{S}^\prime}(d_j+1) 
          + \sum_{v\in V}(f_v+1) + 3 |\mathcal{S}\setminus \mathcal{S}^\prime| + 1  \\
         & = (2p-4)|\mathcal{S}'| + 2 T  + 4 h + |V| + 1. 
        \end{align*}
        Hence $|\mathcal{S}^\prime|\leq k$ holds by the assumption.

        Moreover, we see that $\mathcal{S}^\prime$ is a set cover.
        Suppose not.
        Then there exists an element $v\in V$ such that $\delta (v)\cap \mathcal{S}^\prime = \emptyset$.
        This means that no subset $S_j \in \delta (v)$ possesses  $u_j$ at any point in the reformist sequence.
        So all the agents $x_{j, 0}$ for $S_j \in \delta (v)$ exchange $s_{j, 0}$ with $r_{j, 0}$.
        Since $r_{j, 0}$ receives an envy from the agent $z$, the agent $z$ has to exchange before that.
        However, to exchange items of the agent $z$, we have to exchange items of the agents $v^\ell$'s, but it is impossible before agent $x_{j, 0}$ for some $S_j \in \delta (v)$ exchanges.
        This is a contradiction.
        Thus $\mathcal{S}^\prime$ is a set cover of size $k$.
    \end{proof}

    Let $\mathrm{OPT}$ be the optimal value for the instance of the shortest reformist sequence problem we construct as above.
    Since 
    we may assume that ${\cal S}$ has a set cover of size at least $1$,
    it follows from Claim~\ref{clm:inapprox2} that 
    \begin{equation*}
    \mathrm{OPT}\geq 2p-4 + 2T + 4h + |V|+1.
    \end{equation*}
    Suppose that we can find in polynomial time a reformist sequence of length at most $\alpha \mathrm{OPT}$ steps for some $\alpha \geq 1$.
    By Claim~\ref{clm:inapprox2}, we can construct a set cover of size $k$ where 
    \[
    k \leq \frac{1}{2p-4} \left(\alpha \mathrm{OPT} -2T - 4 h - |V| - 1 \right). 
    \]
    On the other hand, Claim~\ref{clm:inapprox1} implies that an optimal set cover has size at least
    \[
    \frac{1}{2p-4} \left(\mathrm{OPT} -2T - 4 h - |V| - 1 \right). 
    \]
    Hence the approximation ratio for the set cover problem is at most
    \begin{equation*}
    \frac{\alpha \mathrm{OPT} -2T - 4 h - |V| - 1}{\mathrm{OPT} -2T - 4 h - |V| - 1} 
    \leq \alpha \frac{ 2p-4 + 2T + 4h + |V|+1}{2p -4}.
    \end{equation*}
    since the maximum is attained when $\mathrm{OPT}$ is minimum, that is, 
    \begin{equation*}
    \mathrm{OPT} = 2p-4 + 2T + 4h + |V|+1.
    \end{equation*}
    Therefore, if $p$ is sufficiently large, i.e., $2p-4\geq 2T + 4h + |V|+1$, then the approximation ratio is at most $2\alpha$.
    
    We now suppose that $\alpha = c \ln n$ for some sufficiently small constant $c$.
    Since $n=O(h|V|)$ and $h\leq O({|V|}^\beta)$ for some constant $\beta$, the above discussion implies that the set cover problem admits $2c(\beta+1) \ln |V|$-approximation.
    However, this contradicts that the set cover problem is inapproximable within a factor of $(1-\varepsilon)\ln |V|$ for every $\varepsilon >0$.
    This completes the proof.
\end{proof}

On the other hand, the compromise of polynomial-time computability
leads us to fixed-parameter tractability.
In fixed-parameter tractability, we extract a certain
value $k$ from the instance as a \emph{parameter}, and
allow the running time of the form $O(f(k)p(m,n))$, where
$f$ is an arbitrary (but usually computable) function and
$p$ is a polynomial.
An algorithm with such a running time is called a \emph{fixed-parameter algorithm}, and the problem with a
fixed-parameter algorithm is called \emph{fixed-parameter tractable}.

For the shortest reformist sequence problem, we have several
choices of natural parameters.
First, we study the shortest length $\ell$ of a reformist sequence as a parameter.
With this choice, the problem is fixed-parameter tractable.
\begin{theorem}
\label{thm:fpt2}
The decision version of the 
shortest reformist sequence 
problem parameterized by the length $\ell$
of a sequence
is fixed-parameter tractable.
\end{theorem}
\begin{proof}
First, note that after preprocessing in \cref{sec:preprocessing}, if the number $n$ of agents is larger than $\ell$, then there exists
no reformist sequence of length at most $\ell$
since no agent shares an item in the initial envy-free matching and the reformist envy-free matching, and thus the length of every reformist sequence must be at least $n$.
Therefore, after the preprocessing, if $n > \ell$, then
the output is No.

Now, we may assume that $n \leq \ell$.
Then, we consider nominating an arbitrary agent and
exchanging her current item with the best item for her on the table while keeping the envy-freeness.
We iterate nomination at most $\ell$ times.
Then, we obtain a branching algorithm with branching factor $n$ and height $\ell$.
Therefore, the size of the recursion tree is at most
$n^{\ell} \leq \ell^{\ell}$.
Since each exchange can be performed in polynomial time, the whole algorithm runs in $O(\ell^{\ell}p(n,m))$ for some polynomial $p$.
\end{proof}

As the second choice, we study the shortest length 
$\ell$ of a reformist sequence
\emph{minus} the number $n$ of agents as a parameter.
Since the shortest length is at least $n$ (see Section~\ref{sec:preprocessing}),
this parameter can be seen as the number of extra steps needed
to obtain the reformist envy-free matching.

The next theorem shows that it is unlikely to obtain a
fixed-parameter algorithm with this parameter.
Here, W[1]-hardness is a counterpart of NP-hardness in 
fixed-parameter tractability.

\begin{theorem}
\label{thm:w1hard}
It is {\rm W[1]}-hard to decide whether there exists a reformist 
sequence of length at most $n+k$ when $k$ is a parameter.
\end{theorem}
\begin{proof}
    In order to prove the theorem, 
    we reduce the multi-colored clique problem to our problem.
    The multi-colored clique problem is to ask whether, 
    given a $k$-partite graph $G=(V, E)$ with a
    partition $V_1, V_2,\dots, V_k$ of $V$, 
    there exist $k$ vertices $v_1, v_2, \dots, v_k$ such that 
    $v_i\in V_i$ for every integer $i \in \{1,2,\dots,k\}$ 
    and $v_1, v_2, \dots, v_k$ forms a clique.
    It is known~\cite{FHRV09,P03} that the multi-colored clique 
    problem is W[1]-hard when $k$ is a parameter.
    
    Let $G=(V, E)$ with a 
    partition $V_1, V_2,\dots, V_k$ of $V$ 
    be an instance of the multi-colored clique problem.
    We denote $V :=\{v_1, \dots, v_{n^{\prime}}\}$, where $n^{\prime}:=|V|$.
    
    We construct an instance of our problem as follows.
    For each vertex $v_j\in V$, we introduce $d_j+1$ agents
    \begin{equation*}
    v_j^0, v_j^1, \dots, v_j^{d_j},
    \end{equation*}
    where $d_j:=|\delta (v_j)|$ is the degree of $v_j$.
    For each edge $e\in E$, we prepare one agent $e$.
    Moreover, we define one more agent $a$.
    Thus, the set $N$ of agents is 
    \begin{equation*}
    N=\{v_j^\ell\mid v_j\in V, \ \ell \in\{0,1,\dots, d_j\}\}
    \cup E \cup \{a\}.
    \end{equation*}
    The size $|N|$ is equal to 
    \begin{equation*}
    \sum_{v_j\in V} (d_j+1) + |E| + 1 = 
    |V| + 3|E| + 1,
    \end{equation*}
    since $\sum_{v_j\in V} d_j = 2|E|$.
    Define the set $M$ of items to be 
    \begin{equation*}
    M=\{r_i, s_i\mid i\in N\}\cup 
    \{y_e\mid e\in E\}\cup\{z_{v_j}\mid v_j\in V\}.
    \end{equation*}
    The preferences of the agents are defined as follows.
    For each vertex $v_j\in V$, letting 
    $\delta (v_j)=\{e_1, e_2, \dots, e_\ell\}$~(in an arbitrary order), 
    \begin{align*}
        \succ_{v_j^\ell} &\colon r_{v_j^\ell}, y_{e_\ell}, s_{v_j^\ell} \quad (\ell\in \{1,2,\dots, d_j\}),\\
        \succ_{v_j^0} &\colon r_{v_j^0}, z_{v_j}, r_{v_j^d}, r_{v_j^{d-1}}, \dots, r_{v_j^1}, s_{v_j^0}.
    \end{align*}
    For each pair of integers
    $i, j\in \{1,2,\dots, k\}$ such that $i \ne j$,  
    we denote by $E[V_i, V_j]$ the set of 
    edges connecting vertices of $V_i$ and $V_j$.
    Note that every edge belongs 
    to exactly one of the $E[V_i, V_j]$.
    For each pair of integers 
    $i, j\in \{1,2,\dots, k\}$ such that $i\neq j$, 
    denoting 
    $E[V_i, V_j] = \{e_1, \dots, e_t\}$~(in an arbitrary order), 
    \begin{align*}
        \succ_{e_\ell} &\colon r_{e_\ell}, y_{e_\ell}, r_{e_{\ell+1}}, r_a, s_{e_\ell}, \quad (\ell \in \{1,2,\dots, t\}),
    \end{align*}
    where we assume that $e_{t+1}=e_1$.
    The last agent $a$ has preference defined by
    \begin{align*}
        \succ_{a} &\colon r_{a}, r_{v_{n^\prime}^0}, \dots, r_{v_1^0}, s_{a}.
    \end{align*}
    The initial matching $\mu$ is defined to be $\mu(i)=s_i$ for each agent $i$.
    By Claim \ref{clm:W1h1}, a reformist envy-free matching $\sigma$ with respect to $\mu$ is $\sigma(i)=r_{i}$ for each agent $i$.

    \begin{claim}\label{clm:W1h1}
    If $G$ has a multi-colored clique, then there exists a reformist sequence 
    of length $n + \binom{k}{2} + k$.
    \end{claim}
    \begin{proof}
        Let $X$ be a multi-colored clique.
        For simplicity, we re-index vertices so that 
        $X=\{v_1, \dots, v_k\}$ and $v_i\in V_i$ for 
        each integer $i\in\{1,2,\dots, k\}$.
        Consider the following reformist sequence.
        \begin{enumerate}
          \item For each vertex $v_j\in X$, we do the following:
          The agent $v_j^0$ exchanges $s_{v_j^0}$ with $z_{v_j}$.
          Then $v_j^\ell$ exchanges $s_{v_j^\ell}$ with $r_{v_j^\ell}$ for each integer 
          $\ell \in \{1,2,\dots, d_j\}$.
          This takes $d_j+1$ steps for each $v_j\in X$.
          \item For each pair of integers $i, j\in\{1,2,\dots, k\}$ such that $i\neq j$, do the following.
          Let $E[V_i, V_j]=\{e_1,\dots, e_{t_{ij}}\}$ where $t_{ij}:=|E[V_i, V_j]|$.
          As $\{v_i, v_j\}\in E[V_i, V_j]$, we may re-index edges of $E[V_i, V_j]$ so that $e_1=\{v_i, v_j\}$.
          First, the agent $e_1$ exchanges $s_{e_1}$ with $y_{e_1}$.
          This can be done since $y_{e_1}$ has no envy due to Step 1.
          Then in the order of $\ell = 2, 3,\dots, t$, 
          the agent $e_\ell$ exchanges $s_{e_\ell}$ with $r_{e_\ell}$.
          Finally, the agent $e_1$ exchanges $y_{e_1}$ with $r_{e_1}$.
          This takes $t_{ij}+1$ steps for each pair $i, j$.
          Hence the total number of exchanges is $\sum_{i,j} (t_{ij}+1) = |E|+\binom{k}{2}$ as $\sum_{i,j} t_{ij} = |E|$.
          \item The agent $a$ exchanges $s_a$ with $r_a$.
          This can be done because $r_a$ has no envy due to Step 2.
          \item For each vertex $v_j\in X$, the agent $v_j^0$ exchanges $z_{v_j^0}$ with $r_{v_j^0}$.
          This can be done since $r_{v_j^0}$ has no envy from $a$ due to Step 3.
          \item For each vertex $v_j\in V\setminus X$, 
          the agent $v_j^0$ exchanges $s_{v_j^0}$ with $r_{v_j^0}$, and then $v_j^\ell$ exchanges $s_{v_j^\ell}$ with $r_{v_j^\ell}$ for each integer
          $\ell \in \{1,2,\dots, d_j\}$.
        \end{enumerate}
        In the reformist sequence above, each vertex $v_j\in X$ needs $d_j+2$ exchanges in Steps 1 and 4, and each vertex $v_j\in V\setminus X$ requires $d_j+1$ exchanges in Step 5.
        Therefore, since the number of exchanges in Steps 2 and 3 is $|E|+\binom{k}{2}+1$, the total number of exchanges is
        \begin{align*}
          &\sum_{v_j\in X} (d_j+2) + \sum_{v_j\in V\setminus X} (d_j+1) + |E| + \binom{k}{2} + 1 \\
          &= \sum_{v_j\in V} d_j + k + |V|+ |E| + \binom{k}{2} + 1 \\
          &= |V| + 3|E| +\binom{k}{2} +k+ 1
          = |N|+\binom{k}{2}+k,
        \end{align*}
        where the last equation follows from 
        $|N|=|V|+3|E|+1$. This completes the proof. 
    \end{proof}
  
    \begin{claim}\label{clm:W1h2}
    If there exists a reformist sequence of length $|N|+\binom{k}{2}+k$, then $G$ has a multi-colored clique.
    \end{claim}
    \begin{proof}
        Consider a reformist sequence with minimum number of steps.
        We first observe that, because of the minimality, the agent $v_j^\ell$ for each vertex 
        $v_j\in V$ and each integer 
        $\ell \in \{1,2,\dots, d_j\}$ takes only one exchange.
        Moreover, for each pair of integers 
        $i,j\in\{1,2,\dots, k\}$ such that 
        $i \neq j$, the agents in $E[V_i, V_j]$ take $|E[V_i, V_j]|+1$ exchanges in total, and, before that, we need to remove envy at $y_e$ for some $e=\{v_i, v_j\}\in E[V_i, V_j]$ by exchanging $v_i^\ell$'s and $v_j^\ell$'s.
        Also, to exchange with $r_{v_j^0}$ for a vertex $v_j\in V$, we need to remove envy from the agent $a$, which implies that we have to exchange items for the agents $e$ for all $e\in E$ before that.
        
        Define $X$ as the set of $v_j\in V$ such that 
        $v_j^0$ uses $z_{v_j}$ in the reformist sequence.
        Then the number of steps in the sequence is at least $|N|+|X|+\binom{k}{2}$. 
        Since it is at most $|N|+\binom{k}{2}+k$ by the assumption, it follows that $|X|\leq k$.
        We observe that, for any pair of integers 
        $i, j\in\{1,2,\dots, k\}$ such that 
        $i \neq j$, we have $E[V_i, V_j]\cap E[X]\neq \emptyset$, where $E[X]$ is the set of edges induced by $X$.
        In fact, suppose not.
        Then, for such a pair $i, j$ and each edge $e=\{\hat{v}_i, \hat{v}_j\}\in E[V_i, V_j]$, the item $y_e$ receives an envy from some of $\hat{v}_i^{\ell}$'s and $\hat{v}_j^{\ell}$'s.
        Hence we cannot exchange any items on $\bigcup_{e\in E[V_i, V_j]}M_e$.
        This is a contradiction.
        
        Therefore, since $|X|\leq k$, it follows that $X$ forms a multi-colored clique.
  \end{proof}
   
  By Claims~\ref{clm:W1h1} and~\ref{clm:W1h2}, the multi-colored clique problem reduces to the shortest reformist sequence problem, which completes the proof.
\end{proof}

Third, we study the problem parameterized 
by the number of intermediate items.  
Let $K$ denote the set of all the intermediate 
items from the initial envy-free matching $\mu$ to the reformist
envy-free matching $\sigma$, namely, 
$K := M \setminus \{ \mu (i), \sigma(i) \mid i\in N\}$. 
Note that the preprosessing in \cref{sec:preprocessing} does not
increase $|K|$,
and $|K| = m -2n$ holds after the preprocessing. 
We prove that the shortest reformist sequence 
problem parameterized by $|K|$ is fixed-parameter tractable.

\begin{theorem}
\label{thm:fpt}
The shortest reformist sequence 
problem parameterized by $|K|$
is fixed-parameter tractable.
\end{theorem}
\begin{proof}
    In order to prove the theorem, we design a fixed-parameter algorithm.
    \begin{description}
    \item[Step 1.] While some agent $i$ can exchange $\mu (i)$ with $\sigma (i)$, we nominate $i$ to exchange $\mu (i)$ with $\sigma (i)$ and remove $i$ from the instance.
    \item[Step 2.] Choose an item $x \in K$ such that $x\in M_i$ for exactly one agent $i$.
    We solve recursively the following two instances with smaller parameter.
    \begin{itemize}
        \item An instance with the initial matching $\mu'$ where $\mu' (i) = x$ and $\mu' (j) = \mu (j)$ for the other agents $j$
        \item An instance obtained by replacing $M_i$ with $M'_i = M_i \setminus \{x\}$.
    \end{itemize}
    \end{description}
    We first observe that the exchanges in Step 1 can be done first before the other exchanges without destroying the
    minimality of a reformist sequence.
    Thus, we may assume that no agent $i$ can exchange $\mu (i)$ with $\sigma (i)$.
    The next step must be for some agent $i$ to exchange $\mu (i)$ with some item $x\in M_i$.
    Since the resulting matching is envy-free, no agent $j\neq i$ has $x$ in $M_{j}$.
    
    We consider branching using such an item $x$.
    That is, we pick arbitrarily an item $x \in K$ such that $x\in M_i$ for exactly one agent $i$, and consider two cases: when the next step is to exchange $\mu(i)$ with $x$ for the agent $i$, or when the item $x$ is never used in the reformist sequence.
    We note that, if $x$ is used in the reformist sequence, then we can exchange $\mu(i)$ with $x$ now before the other agents' exchange, as it does not worse the situation.
    For the former one, we consider the instance with the initial matching $\mu'$ where $\mu' (i) = x$ and $\mu' (j) = \mu (j)$ for the other agents $j$.
    For the latter one, we solve the instance obtained by replacing $M_i$ with $M'_i = M_i \setminus \{x\}$.
    For each case, the parameter $|K|$ is decreased by one.
    Thus the depth of recursion is at most $|K|$.
    Therefore, the total time complexity is $O(2^{|K|}p(m,n))$
    for some polynomial $p$.
\end{proof}

\section{Conclusion}

We studied a process of iterative improvement of envy-free matchings in
the house allocation problem, and defined a reformist envy-free
matching as an outcome of the process.
We proved that a reformist envy-free matching is unique up to 
the choice of an initial envy-free matching.
Then, we studied the shortest reformist sequence problem and showed a
contrast between NP-hardness and polynomial-time solvability with
respect to the lengths of preference lists of agents and the number of
occurrences of each item in the preference lists.

Several questions remain unsolved.
As for approximation, we proved the inapproximability of factor 
$c \ln n$ for some constant $c$.
On the other hand, we do not know any approximation algorithm.
As for fixed-parameter tractability, we showed an
fixed-parameter algorithm when the length of a reformist sequence or
the number of intermediate items is a parameter.
On the other hand, we do not know this is also the case when $n$ is a
parameter.
Another direction of research may look at the case where preferences
may contain a tie or a pair of incomparable items.

\bibliographystyle{plain}
\bibliography{envyfree}

\begin{thebibliography}{10}

\bibitem{AzizM20}
Haris Aziz and Simon Mackenzie.
\newblock A bounded and envy-free cake cutting algorithm.
\newblock {\em Communications of the {ACM}}, 63(4):119--126, 2020.

\bibitem{BCGLMW18}
Aur{\'{e}}lie Beynier, Yann Chevaleyre, Laurent Gourv{\`{e}}s, Ararat
  Harutyunyan, Julien Lesca, Nicolas Maudet, and Ana{\"{e}}lle Wilczynski.
\newblock Local envy-freeness in house allocation problems.
\newblock {\em Autonomous Agents and Multi-Agent Systems}, 33(5):591--627,
  2019.

\bibitem{BouveretCM16}
Sylvain Bouveret, Yann Chevaleyre, and Nicolas Maudet.
\newblock Fair allocation of indivisible goods.
\newblock In Felix Brandt, Vincent Conitzer, Ulle Endriss, J{\'{e}}r{\^{o}}me
  Lang, and Ariel~D. Procaccia, editors, {\em Handbook of Computational Social
  Choice}, pages 284--310. Cambridge University Press, Cambridge, UK, 2016.

\bibitem{BW19}
Felix Brandt and Ana{\"{e}}lle Wilczynski.
\newblock On the convergence of swap dynamics to pareto-optimal matchings.
\newblock In Ioannis Caragiannis, Vahab~S. Mirrokni, and Evdokia Nikolova,
  editors, {\em Proceedings of the 15th Conference on Web and Internet
  Economics Web and Internet Economics}, volume 11920 of {\em Lecture Notes in
  Computer Science}, pages 100--113, Cham, Switzerland, 2019. Springer.

\bibitem{ChaudhuryGM20}
Bhaskar~Ray Chaudhury, Jugal Garg, and Kurt Mehlhorn.
\newblock {EFX} exists for three agents.
\newblock In P{\'{e}}ter Bir{\'{o}}, Jason~D. Hartline, Michael Ostrovsky, and
  Ariel~D. Procaccia, editors, {\em Proceedings of the 21st {ACM} Conference on
  Economics and Computation}, pages 1--19, New York, NY, 2020. {ACM}.

\bibitem{DS14}
Irit Dinur and David Steurer.
\newblock Analytical approach to parallel repetition.
\newblock In David~B. Shmoys, editor, {\em Proceedings of the 46th Annual ACM
  Symposium on Theory of Computing}, pages 624--633, New York, NY, 2014. {ACM}.

\bibitem{FHRV09}
Michael~R. Fellows, Danny Hermelin, Frances~A. Rosamond, and St{\'{e}}phane
  Vialette.
\newblock On the parameterized complexity of multiple-interval graph problems.
\newblock {\em Theoretical Computer Science}, 410(1):53--61, 2009.

\bibitem{GSV19}
Jiarui Gan, Warut Suksompong, and Alexandros~A. Voudouris.
\newblock Envy-freeness in house allocation problems.
\newblock {\em Mathematical Social Sciences}, 101:104--106, 2019.

\bibitem{GoldbergHS20}
Paul Goldberg, Alexandros Hollender, and Warut Suksompong.
\newblock Contiguous cake cutting: Hardness results and approximation
  algorithms.
\newblock {\em Journal of Artificial Intelligence Research}, 69:109--141, 2020.

\bibitem{GLW17}
Laurent Gourv{\`{e}}s, Julien Lesca, and Ana\"{e}lle Wilczynski.
\newblock Object allocation via swaps along a social network.
\newblock In Carles Sierra, editor, {\em Proceedings of the 26th International
  Joint Conference on Artificial Intelligence}, pages 213--219, Palo Alto, CA,
  2017. {AAAI} Press.

\bibitem{HX20}
Sen Huang and Mingyu Xiao.
\newblock Object reachability via swaps under strict and weak preferences.
\newblock {\em Autonomous Agents and Multi-Agent Systems}, 34(2):51, 2020.

\bibitem{IDHPSUU11}
Takehiro Ito, Erik~D. Demaine, Nicholas J.~A. Harvey, Christos~H.
  Papadimitriou, Martha Sideri, Ryuhei Uehara, and Yushi Uno.
\newblock On the complexity of reconfiguration problems.
\newblock {\em Theoretical Computer Science}, 412(12-14):1054--1065, 2011.

\bibitem{K72}
Richard~M. Karp.
\newblock Reducibility among combinatorial problems.
\newblock In Raymond~E. Miller and James~W. Thatcher, editors, {\em Proceedings
  of a symposium on the Complexity of Computer Computations}, The {IBM}
  Research Symposia Series, pages 85--103. Plenum Press, 1972.

\bibitem{KlausMR16}
Bettina Klaus, David~F. Manlove, and Francesca Rossi.
\newblock Matching under preferences.
\newblock In Felix Brandt, Vincent Conitzer, Ulle Endriss, Jer\^{o}me Lang, and
  Ariel~D. Procaccia, editors, {\em Handbook of Computational Social Choice},
  pages 333--355. Cambridge University Press, Cambridge, UK, 2016.

\bibitem{KrishnaaLNN20}
Prem Krishnaa, Girija Limaye, Meghana Nasre, and Prajakta Nimbhorkar.
\newblock Envy-freeness and relaxed stability: Hardness and approximation
  algorithms.
\newblock In Tobias Harks and Max Klimm, editors, {\em Proceedings of the 13th
  Symposium on Algorithmic Game Theory}, volume 12283 of {\em Lecture Notes in
  Computer Science}, pages 193--208. Springer, 2020.

\bibitem{M13}
David~F. Manlove.
\newblock {\em Algorithmics of Matching Under Preferences}.
\newblock World Scientific, Singapore, 2013.

\bibitem{N18}
Naomi Nishimura.
\newblock Introduction to reconfiguration.
\newblock {\em Algorithms}, 11(4):52, 2018.

\bibitem{P03}
Krzysztof Pietrzak.
\newblock On the parameterized complexity of the fixed alphabet shortest common
  supersequence and longest common subsequence problems.
\newblock {\em Journal of Computer and System Sciences}, 67(4):757--771, 2003.

\bibitem{Procaccia16}
Ariel~D. Procaccia.
\newblock Cake cutting algorithms.
\newblock In Felix Brandt, Vincent Conitzer, Ulle Endriss, J{\'{e}}r{\^{o}}me
  Lang, and Ariel~D. Procaccia, editors, {\em Handbook of Computational Social
  Choice}, pages 311--330. Cambridge University Press, Cambridge, UK, 2016.

\bibitem{WuR18}
Qingyun Wu and Alvin~E. Roth.
\newblock The lattice of envy-free matchings.
\newblock {\em Games and Economic Behavior}, 109:201--211, 2018.

\bibitem{Yokoi20}
Yu~Yokoi.
\newblock Envy-free matchings with lower quotas.
\newblock {\em Algorithmica}, 82(2):188--211, 2020.

\end{thebibliography}


%
%
%
%
%


%
%

%
%
%
%
%
%
%
%
%
%
%
%
%

%
%
%
%
%
%
%
%
%
%
%

%
%
%
%
%
%

%
%
%
%
%
%
%
%
%
%
%
%
%

%
%
%
%
%
%
%
%
%
%
%
%

%
%
%
%
%
%

%

%
%
%
%

%
%
%
%
%
%

%
%
%
%
%
%

%
%
%
%
%
%
%

\end{document}